\documentclass[a4paper,11pt,english]{amsart}
\usepackage{microtype}
\usepackage{amsthm}
\usepackage[foot]{amsaddr}
\usepackage{graphicx} \usepackage{array} \usepackage{colortbl,makecell}
\usepackage{pifont} 

\usepackage{amsmath, amssymb, amsfonts, verbatim}
\usepackage{hyphenat, epsfig, subcaption, multirow}
\usepackage{nicefrac}
\usepackage{paralist}
\usepackage[shortlabels]{enumitem}

\usepackage{dsfont} 

\usepackage{mathtools, etoolbox}
\usepackage{thmtools}
\usepackage{thm-restate}
\usepackage{xparse}

\usepackage{wrapfig}

\usepackage[noend]{algorithmic}
\usepackage{algorithm}

\usepackage[font=small,labelfont=bf]{caption}

\usepackage{booktabs,tabularx,multirow}

\usepackage[dvipsnames]{xcolor}

\usepackage{babel}
\usepackage{csquotes}

\usepackage[style=alphabetic,backref=true,maxnames=99,maxalphanames=99,
            isbn=false,doi=true]{biblatex}

\AtEveryBibitem{\clearname{editor}}
\AtEveryBibitem{\clearlist{publisher}}
\AtEveryBibitem{\clearfield{pages}}
\DeclareFontFamily{U}{mathx}{\hyphenchar\font45}
\DeclareFontShape{U}{mathx}{m}{n}{
      <5> <6> <7> <8> <9> <10>
      <10.95> <12> <14.4> <17.28> <20.74> <24.88>
      mathx10
      }{}
\DeclareSymbolFont{mathx}{U}{mathx}{m}{n}
\DeclareMathSymbol{\bigtimes}{1}{mathx}{"91}

\usepackage{tcolorbox}
\tcbuselibrary{skins,breakable}
\tcbset{enhanced jigsaw}

\definecolor{DarkRed}{rgb}{0.5,0.1,0.1}
\definecolor{DarkBlue}{rgb}{0.1,0.1,0.5}

\usepackage{nameref}
\definecolor{ForestGreen}{rgb}{0.1333,0.5451,0.1333}
\definecolor{Red}{rgb}{0.9,0,0}
\usepackage[linktocpage=true,
	colorlinks,
	linkcolor=DarkRed,citecolor=ForestGreen,
    bookmarks,bookmarksopen,bookmarksnumbered]
	{hyperref}
\usepackage[noabbrev,nameinlink,capitalize]{cleveref}
\crefname{property}{property}{Property}
\creflabelformat{property}{(#1)#2#3}
\crefname{equation}{Eq}{Eq}
\creflabelformat{equation}{(#1)#2#3}

\usepackage{bm}
\usepackage{url}
\usepackage{xspace}
\usepackage[mathscr]{euscript}
\usepackage{mathrsfs}

\usepackage{tikz}
\usetikzlibrary{arrows}
\usetikzlibrary{arrows.meta}
\usetikzlibrary{shapes}
\usetikzlibrary{backgrounds}
\usetikzlibrary{positioning}
\usetikzlibrary{decorations.markings}
\usetikzlibrary{patterns}
\usetikzlibrary{calc}
\usetikzlibrary{fit}
\tikzset{vertex/.style={circle, black, fill=Yellow, line width=1pt, draw, minimum width=8pt, minimum height=8pt, inner sep=0pt}}

\usepackage[framemethod=TikZ]{mdframed}

\usepackage[margin=1in]{geometry}

\usepackage{soul}

\renewcommand{\paragraph}[1]{\medskip\noindent\textbf{#1}}

\newtheorem{theorem}{Theorem}
\newtheorem{lemma}{Lemma}[section]
\newtheorem{proposition}[lemma]{Proposition}

\newtheorem*{claim*}{Claim}
\newtheorem*{proposition*}{Proposition}
\newtheorem*{lemma*}{Lemma}
\newtheorem*{problem*}{Problem}

\crefname{lemma}{Lemma}{Lemmas}
\crefname{claim}{Claim}{Claims}
\crefname{enumi}{Step}{Steps}
\crefname{step}{Step}{Steps}

\theoremstyle{definition}
\newtheorem{definition}[lemma]{Definition}

\theoremstyle{definition}

\newenvironment{abox}{\begin{tcolorbox}[
		enlarge top by=5pt,
		enlarge bottom by=5pt,
		breakable,
		frame hidden,
		overlay broken = {
			\draw[line width=1pt, black]
			(frame.north west) rectangle (frame.south east);},
		overlay = {
			\draw[line width=1pt, black]
			(frame.north west) rectangle (frame.south east);},
		boxsep=0pt,
		left=4pt,
		right=4pt,
		top=10pt,
		arc=0pt,
		boxrule=1pt,toprule=1pt,
		colback=white
	]}
{\end{tcolorbox}}

\newtheorem{mdalg}{Algorithm}
\newenvironment{Algorithm}{\begin{abox}\begin{mdalg}}{\end{mdalg}\end{abox}}

\renewcommand{\qed}{\nobreak \ifvmode \relax \else
      \ifdim\lastskip<1.5em \hskip-\lastskip
      \hskip1.5em plus0em minus0.5em \fi \nobreak
      \vrule height0.75em width0.5em depth0.25em\fi}

\renewcommand{\leq}{\leqslant}
\renewcommand{\geq}{\geqslant}

\DeclarePairedDelimiter{\bracket}[]
\DeclarePairedDelimiter{\paren}()

\DeclarePairedDelimiter{\floor}{\lfloor}{\rfloor}

\DeclarePairedDelimiter{\set}{\{}{\}}

\DeclarePairedDelimiterXPP{\Ot}[1]{\widetilde{O}}(){}{#1}
\DeclarePairedDelimiterXPP{\Omgt}[1]{\widetilde{\Omega}}(){}{#1}
\DeclarePairedDelimiterXPP{\BigO}[1]{O}(){}{#1}

\DeclareMathOperator{\Trace}{Tr}
\DeclarePairedDelimiterXPP{\tr}[1]{\Trace}(){}{#1}

\NewDocumentCommand{\Prob}{sO{}E{_}{{}}m}{{\boldsymbol{\mathbb{P}}}_{#3}
  \IfBooleanTF{#1}
  {\bracket*{#4}}
  {\bracket[#2]{#4}}
}

\NewDocumentCommand{\Exp}{sO{}E{_}{{}}m}{\expect_{#3}
  \IfBooleanTF{#1}
  {\bracket*{#4}}
  {\bracket[#2]{#4}}
}

\renewcommand{\epsilon}{\varepsilon}

\newcommand{\poly}{\mbox{\rm poly}}

\newenvironment{tbox}{\begin{tcolorbox}[
		enlarge top by=5pt,
		enlarge bottom by=5pt,
		 breakable,
		 boxsep=0pt,
                  left=4pt,
                  right=4pt,
                  top=10pt,
                  arc=0pt,
                  boxrule=1pt,toprule=1pt,
                  colback=white
                  ]}
{\end{tcolorbox}}

\newcommand{\qLOCAL}{quantum-LOCAL\xspace}
\newcommand{\qPN}{quantum-PN\xspace}

\newcommand{\bN}{\mathbb{N}}

\newcommand{\bR}{\mathbb{R}}

\newcommand{\ket}[1]{|#1\rangle}
\newcommand{\bra}[1]{\langle#1|}
\newcommand{\braket}[2]{\langle #1|#2\rangle}
\newcommand{\ketbra}[2]{|#1\rangle\langle#2|}
\newcommand{\braUket}[3]{\langle #1|#2|#3 \rangle}

\newcommand{\bC}{\mathbb{C}}
\DeclareMathOperator{\Mat}{Mat}
\DeclareMathOperator{\im}{im}
\DeclareMathOperator{\End}{End}
\newcommand{\rL}{\mathsf{L}}
\newcommand{\rM}{\mathsf{M}}
\newcommand{\rR}{\mathsf{R}}
\newcommand{\rX}{\mathsf{X}}
\newcommand{\rY}{\mathsf{Y}}
\newcommand{\rZ}{\mathsf{Z}}
\newcommand{\rT}{\mathsf{T}}
\newcommand{\rE}{\mathsf{E}}

\newcommand{\rSigma}{\mathsf{\Sigma}}

\usepackage{dsfont}
\newcommand{\Id}{\mathds{1}}
\newcommand{\One}{\mathds{1}}

\newcommand{\algA}{\mathscr{A}}
\newcommand{\algB}{\mathscr{B}}

\title{Distributed Quantum Algorithms Cannot Color Cycles with Probability 1}
\date{}

\author{Xavier Coiteux-Roy\textsuperscript{1}}
\author{Maxime Flin\textsuperscript{2}}
\author{Carlos de Gois\textsuperscript{3}}
\author{Marc-Olivier Renou\textsuperscript{3}}
\author{Jukka Suomela\textsuperscript{2}}
\author{Isadora Veeren\textsuperscript{3}}

\address{\textsuperscript{1}\,Institute for Quantum Science and Technology, University of Calgary, Canada}
\address{\textsuperscript{2}\,Aalto University, Finland}
\address{\textsuperscript{3}\,Inria Paris-Saclay, CPHT, CNRS, Ecole Polytechnique, Institut Polytechnique de Paris, France}

\begin{document}

\begin{abstract}
    We prove that any \emph{distributed quantum algorithm} that finds a $3$-coloring with probability $1$ in a cycle of anonymous identical computers has to be global, that is, it needs $\Omega(n)$ communication rounds. It follows that quantum computation and communication does not help with this problem.

    All prior lower bounds on quantum advantage in distributed graph algorithms use arguments related to physical causality. However, it is known that such arguments cannot rule out fast quantum advantage for $3$-coloring cycles. In particular, any causality-based argument would rule out the existence of \emph{finitely dependent coloring}, but Holroyd and Liggett (2016) showed that such colorings do exist.
    
    Hence to tackle this problem, we need a ``genuinely quantum'' lower-bound technique that can distinguish between (1)~distributions that do not violate physical causality vs.\ (2)~distributions that can be realized with a quantum strategy. We present the first such lower-bound technique in this context. First, we show that $1$-round quantum algorithms cannot break symmetry with probability $1$. Second, we present a \emph{wishful teleportation} strategy that can be used to turn $T$-round quantum 3-coloring algorithms into $1$-round quantum algorithms breaking symmetry, while preserving success probability $1$. Put together, the lower bound follows.
\end{abstract}

\maketitle

\section{Introduction}
\label{sec:intro}

In this work, we present the first ``genuinely quantum'' lower bound that puts limits on quantum advantage in distributed symmetry breaking and coloring of cycle. In particular, we show that no distributed quantum algorithm can color cycles with probability $1$ in a sublinear number of rounds.
Previous lower bounds were established only for \emph{finitely dependent distributions}, a computational model that contains distributed quantum algorithms while remaining easier to analyse. However, finitely dependent colorings of cycles are known to exist \cite{holroyd-liggett-2016-finitely-dependent-coloring}.
Hence, our work is the first inherently quantum lower bound that separates distributed quantum algorithms from finitely dependent distributions.

\subsection{Quantum Advantage for Distributed Symmetry Breaking}

One of the biggest open questions related to distributed quantum advantage is whether quantum computation and communication help with distributed symmetry-breaking problems \cite{gavoille-kosowski-markiewicz-2009-what-can-be-observed,akbari-coiteux-roy-etal-2025-online-locality-meets,suomela-2024-open-problems-related-to-locality-in,d-amore-2025-on-the-limits-of-distributed-quantum}. The canonical symmetry-breaking primitive is $3$-coloring directed paths and cycles; if this problem can be solved efficiently, all other symmetry-breaking tasks can also be solved efficiently.

The main formalism where this question has been studied is the LOCAL model of distributed computing \cite{linial-1992-locality-in-distributed-graph-algorithms,peleg-2000-distributed-computing-a-locality-sensitive}: we have a cycle formed by $n$ computers; they are initially labeled with $\poly(n)$-sized unique identifiers; computation proceeds in synchronous rounds; in each round, each node can send a message to each neighbor and update its state based on the messages it receives; and eventually, after some $T$ rounds, each node must stop and announce its own color. In this classical deterministic setting, the complexity of the task is very well-understood since the 1990s: $\Theta(\log^* n)$ communication rounds is sufficient \cite{cole-vishkin-1986-deterministic-coin-tossing-with,goldberg-plotkin-shannon-1988-parallel-symmetry} and necessary \cite{linial-1992-locality-in-distributed-graph-algorithms}. Even if we allow randomness and only require a correct solution with high probability (w.h.p., with probability $1-1/\poly(n)$), it is known that $\Theta(\log^* n)$ rounds are still necessary \cite{naor-1991-a-lower-bound-on-probabilistic-algorithms-for}.

But does anything change if we switch from classical LOCAL to the quantum-LOCAL model \cite{gavoille-kosowski-markiewicz-2009-what-can-be-observed,arfaoui-fraigniaud-2014-what-can-be-computed-without}, where each node is a quantum computer and can exchange qubits with its neighbors? While we have seen a lot of progress in understanding which distributed problems admit quantum advantage and which problems remain hard in quantum-LOCAL \cite{gavoille-kosowski-markiewicz-2009-what-can-be-observed,arfaoui-fraigniaud-2014-what-can-be-computed-without,le-gall-nishimura-rosmanis-2019-quantum-advantage-for,dhar-kujawa-etal-2024-local-problems-in-trees-across-a,akbari-coiteux-roy-etal-2025-online-locality-meets,balliu-brandt-etal-2025-distributed-quantum-advantage,balliu-casagrande-etal-2026-distributed-quantum,balliu-coupette-etal-2025-new-limits-on-distributed,brandt-gottlicher-2026-a-post-quantum-lower-bound-for}, the seemingly elementary question of $3$-coloring cycles has remained open, even though essentially the same question was already posed in the very first paper that introduced the quantum-LOCAL model \cite{gavoille-kosowski-markiewicz-2009-what-can-be-observed}.
Nobody has been able to rule out the possibility that there is, for example, a quantum-LOCAL algorithm that $3$-colors cycles in a constant number $O(1)$ of rounds of communication.

\subsection{Barrier: Finitely Dependent Colorings}

There is a barrier that has so far prevented progress on this question. To our knowledge, all known lower bounds in the quantum-LOCAL model are \emph{indirect}: they do not directly study the quantum-LOCAL model (or any model of quantum computing), but they study  strictly stronger models,
such as finitely dependent distributions \cite{burton-goulet-meester-1993-on-1-dependent-processes-and,holroyd-liggett-2016-finitely-dependent-coloring,akbari-coiteux-roy-etal-2025-online-locality-meets}, non-signaling distributions \cite{gavoille-kosowski-markiewicz-2009-what-can-be-observed,arfaoui-fraigniaud-2014-what-can-be-computed-without}, and the online-LOCAL model \cite{akbari-eslami-etal-2023-locality-in-online-dynamic,akbari-coiteux-roy-etal-2025-online-locality-meets}. These lower bounds fundamentally build on the idea that distributed quantum algorithms cannot violate physical causality---information must be carried by communicated systems and correlated event must share a common cause.

Unfortunately, this avenue is blocked for the case of symmetry-breaking problems, and especially for the case of $3$-coloring cycles. It is known that there exists a finitely dependent distribution of $3$-colorings in cycles \cite{holroyd-liggett-2016-finitely-dependent-coloring}. It follows that a hypothetical constant-round quantum algorithm for $3$-coloring cycles would not violate causality, and hence purely causality-based arguments cannot be used.

Therefore, any proof that shows the impossibility of efficient $3$-coloring in quantum-LOCAL must be ``genuinely quantum'', in the sense that it cannot as a byproduct rule out finitely dependent distributions, non-signaling distributions, or online-LOCAL algorithms.

We point out that there is a paper that rules out {one-round, one-sided algorithms} for $3$-coloring \cite{le-gall-rosmanis-2022-non-trivial-lower-bound-for-3,gavoille-kachigar-zemor-2019-localisation-resistant} and, very recently, a paper proving that 3-coloring $\Delta$-ary trees with $\Delta \geq 300$ requires $\Omega(\log^* \Delta)$ rounds \cite{fraigniaud-magniez-ziccardi-2026-no-distributed-quantum}.
Both lower bounds are about the non-signaling model, i.e., causality-based---fundamentally different techniques are needed for stronger lower bounds.

\subsection{Our Contribution}

In this work, we present the \emph{first ``genuinely quantum'' lower bound for the task of distributed $3$-coloring}; our lower bound \emph{separates quantum algorithms from finitely dependent distributions}.
We prove that there is no $o(n)$-round quantum algorithm using finitely many qubits in anonymous cycles that produces a $3$-coloring with probability $1$; the precise statement is as follows:
\begin{theorem}
  \label{thm:coloring}
  For every fixed $c\geq 3$ and every even $n \geq 4$, any \qPN algorithm that properly
  $c$-colors directed $n$-cycles with probability $1$ takes $T(n) > \floor{n/2} - 2$ rounds.
\end{theorem}
\noindent We note that our result comes with two assumptions:
\begin{enumerate}
    \item We work in the \qPN model; this is the \emph{anonymous} version of quantum-LOCAL (with finitely many qubits). In particular the nodes do not have unique identifiers.
    \item We only rule out algorithms that work with \emph{probability $1$}; our result does not rule out algorithms that work w.h.p.
\end{enumerate}
We argue that, while somewhat nonstandard, this setting is still relevant and nontrivial:
\begin{enumerate}
    \item While probability-$1$ classical algorithms do not usually make much sense, there are several examples of nontrivial distributed quantum algorithms that indeed guarantee success with probability $1$ using only finitely many qubits; examples include the quantum algorithm for distributed symmetry-breaking in two-node networks \cite{gidney-2014-perfect-symmetry-breaking-with-quantum} and the solution strategy for the GHZ game \cite{greenberger-horne-zeilinger-1989-going-beyond-bell-s}.
    \item The finitely dependent colorings \cite{holroyd-liggett-2016-finitely-dependent-coloring} are also error-free, and hence a direct distributed analogue would be a quantum algorithm with success probability $1$. Moreover, they are invariant under shifts, which would naturally correspond to a quantum algorithm in the anonymous setting.
\end{enumerate}
\cref{tab:results} gives an overview of what is now known about the $3$-coloring problem across various models of computing.

\begin{table}
    \definecolor{newcolor}{HTML}{f26924}
    \definecolor{opencolor}{HTML}{0088cc}
    \centering
    \begin{tabular}{llll}
        \toprule
        Model & With probability 1 & With high probability \\
        \midrule
        deterministic PN & impossible & impossible \\
        randomized PN & $\Theta(n)$ & $\Theta(\log^* n)$ \\
        quantum-PN & \textcolor{newcolor}{$\Theta(n)$ -- \textbf{\em new}} & \textcolor{opencolor}{$O(\log^* n)$} \\
        \midrule
        deterministic LOCAL & $\Theta(\log^* n)$ & $\Theta(\log^* n)$ \\
        randomized LOCAL & $\Theta(\log^* n)$ & $\Theta(\log^* n)$ \\
        quantum-LOCAL & \textcolor{opencolor}{$O(\log^* n)$} & \textcolor{opencolor}{$O(\log^* n)$} \\
        \midrule
        finitely dependent & $\Theta(1)$ & $\Theta(1)$ \\
        non-signaling & $\Theta(1)$ & $\Theta(1)$ \\
        SLOCAL & $\Theta(1)$ & $\Theta(1)$ \\
        online-LOCAL & $\Theta(1)$ & $\Theta(1)$ \\
        \bottomrule
    \end{tabular}
    \caption{What is now known about the problem of $3$-coloring cycles across various models of computing. 
The values are round complexities or localities; the nontrivial results are from \cite{cole-vishkin-1986-deterministic-coin-tossing-with,linial-1992-locality-in-distributed-graph-algorithms,naor-1991-a-lower-bound-on-probabilistic-algorithms-for,holroyd-liggett-2016-finitely-dependent-coloring}. We have highlighted the \textcolor{newcolor}{new result} and the \textcolor{opencolor}{open questions} with colors.}
        \label{tab:results}
\end{table}

\subsection{Open Questions}

After this work, the main open question is extending our result from algorithms that work with probability $1$ to algorithms that work w.h.p. We believe that this is the only essential missing part before we can close the long-standing open question. As soon as we look at algorithms that work w.h.p., the distinction between anonymous networks and networks with unique identifiers essentially disappears (since random identifiers will be unique w.h.p.).

\section{Technical Overview}
\label{sec:technical-overview}

The proof of \cref{thm:coloring} is in two steps: first, we prove that in any 1-round \qPN algorithm, there is a non-zero probability that $n-2$ consecutive nodes output the same value; then, we prove that any zero-error $T$-round coloring algorithm can be turned into a 1-round algorithm (using more colors) that is breaking symmetry with probability one, meaning that no $2T + 2$ consecutive nodes all output the same value. 
Hence, any zero-error coloring algorithm must use $T > n/2-1$ rounds, as if not, it could be used to produce a 1-round \qPN algorithm in which any $2T + 2 \leq n-2$ consecutive nodes never all output the same value.

This work was assisted by modern AI tools: the key idea of the first part, the one-round lower bound, was discovered by OpenAI GPT 5.4, while the key idea of part (2) was discovered by us. However, we emphasize that this paper was fully written (and verified) by humans, for humans.

\paragraph{The \qPN model (\cref{sec:background}).}
Since nodes are in a cycle, we identify them using integers from 1 to n understood cyclically modulo $n$. In the \qPN model, communications rounds are seen as swap gates in a quantum circuit, see \cref{fig:example-quantum-LOCAL}.
Each node manipulates three local registers, local computation steps correspond to unitary matrices, and that the output colors are obtained by a final projective measurement. 
Using a standard Stinespring dilation, by increasing the dimension of the local registers, one can verify that this model captures the more general definition of \qPN in which unitary matrices can be any local quantum transformation (i.e., any CPTP map).

\begin{figure}[ht!]
    \centering
    \includegraphics[width=.5\linewidth]{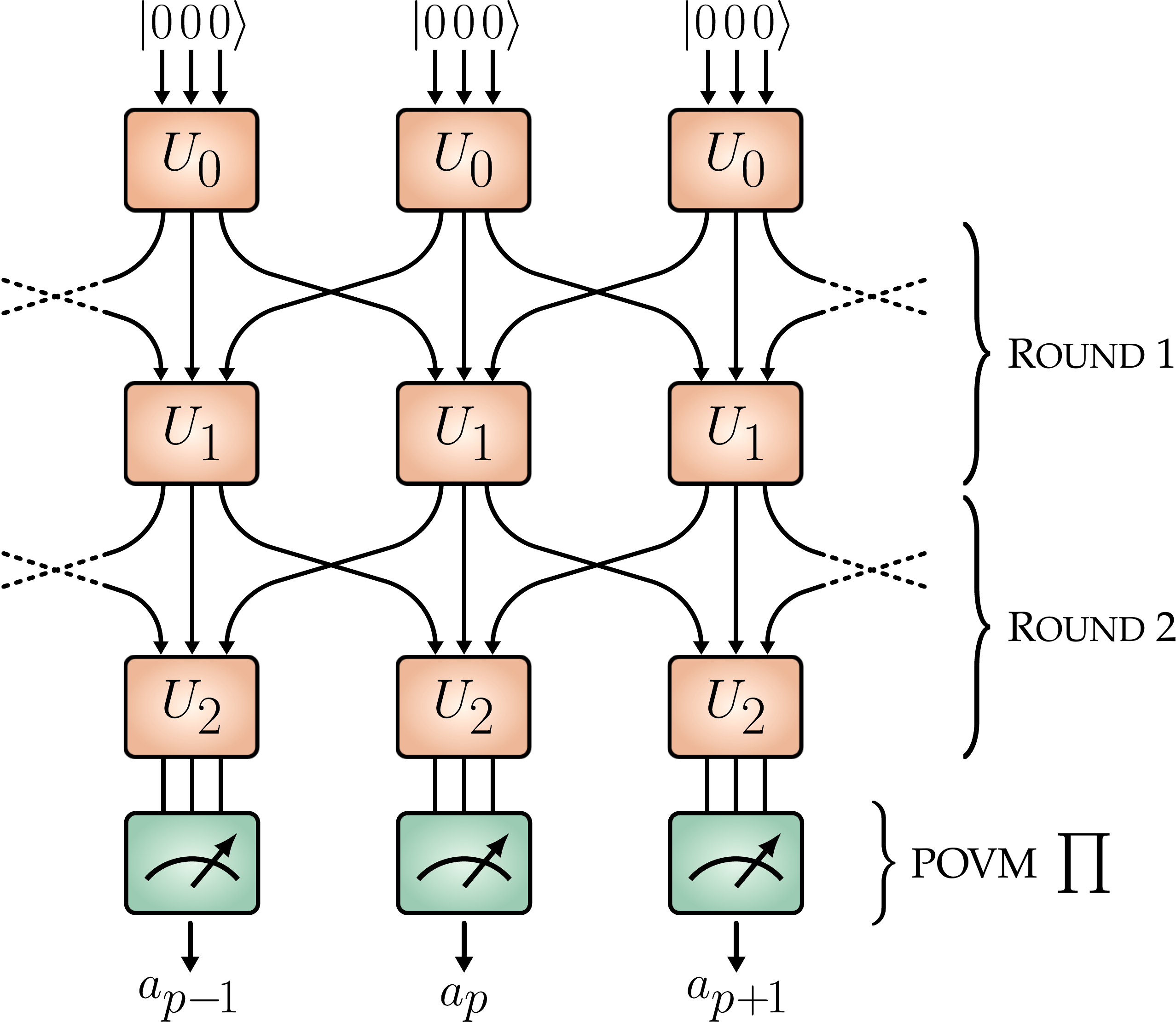}
    \caption{The normal form of a 2-round \qPN algorithm. It interleaves local computations, in the form of unitary matrices $U_0$, $U_1$, and $U_2$ and register swaps with the left and right neighbors. The final outcome is the result of a measurement (formalized as a POVM). See \cref{sec:background} for formal definitions.\label{fig:example-quantum-LOCAL}}
\end{figure}

\paragraph{The One-Round Lower Bound (\cref{sec:one-round}).}
The proof proceeds by contradiction. Assume that there exists a one-round, identical nodes quantum algorithm which breaks symmetry, i.e., such that the probability that any $n-2$ consecutive nodes output the same color is zero. First, we express the global state after the communication round as a translation-invariant Matrix Product State (MPS) \cite{fannes-nachtergaele-werner-1992-finitely-correlated,perez-garcia-verstraete-etal-2007-matrix-product-state,cirac-perez-garcia-etal-2021-matrix-product-states-and}.
Two structural properties of this MPS are then derived.
One follows from the symmetry-breaking requirement, while the other follows from the fact that the algorithm uses a single communication round.
Finally, we show that these two properties are incompatible.

For the sake of intuition, we consider here a simplified case where the number of nodes $n$ is even, the number of colors is equal to the dimension of each local quantum state and each node produces its output by a measurement in the computation basis.
Using MPS, the global state $\ket{\psi}$ of the network after the communication round can be written as
\[
\ket{\psi}
=
\sum_{x_1,\ldots,x_n\in[c]^n}
\Trace\!\left(
\prod_{p=1}^{n} L(x_p)
\right)
\ket{x_1,\ldots,x_n},
\]
where $L$ is a linear map from local vector states to matrices. Since $\ket{\psi}$ is translationally invariant, the map $L$ is independent of the node index.

Next, the assumption that the algorithm never output colors $x,\ldots,x,y,z$ implies that for all $x,y,z\in[c]$, $\braket{x,\ldots,x,y,z}{\psi}=0$, which in the MPS representation reads
\begin{equation}
\label{eq:intro-zero-trace}
\Trace\!\left(L(x)^{n-2}L(y)L(z)\right)=0
\qquad
\text{for all } x,y,z\in\mathbb [c] .
\end{equation}

We then show (see \cref{lem:block}) that every state generated by a one-round algorithm must be such that, after a suitable change of basis,
\begin{equation} \label{eq:intro-block} 
L(x) = 
\left[ \begin{array}{cccc} 
M_x & 0 & * & 0 \\
* & 0 & * & 0 \\ 
0 & 0& 0 & 0 \\ 
0 & 0& 0 & 0 
\end{array} \right] 
\quad\text{and thus}\quad 
L(x)^{n-2} = 
\left[ \begin{array}{cccc} 
M_x^{n-2} & 0 & * & 0 \\ 
* & 0 & * & 0 \\ 
0 & 0& 0 & 0 \\ 
0 & 0& 0 & 0 
\end{array} \right]\ , 
\end{equation}
where $M_x$ is some matrix that depends (linearly) on $x$, and $\dim(M) \leq c$.

Finally, we prove that \cref{eq:intro-zero-trace} and \cref{eq:intro-block} cannot simultaneously hold. Combining these two conditions implies that $\Trace(L(x))=0$ for every $x\in [c]$. On the other hand, we can find $x$ such that $\tr{M_x} \neq 0$, which contradicts $\tr{L(x)} = 0$.

\paragraph{Round Reduction (\cref{sec:round-reduction}).}
We now give an informal overview of the reduction from $T$-round coloring algorithms to one-round symmetry breaking. Suppose that there exists a \qPN algorithm $\algA$ which produces a proper coloring in $T$ communication rounds with probability $1$. We show how to transform $\algA$ into a one-round \qPN algorithm $\algB$ which breaks symmetry on every block of $2T+2$ consecutive nodes, meaning that such a block can never produce identical outputs.
The idea is for $\algB$ to simulate $\algA$ using only one communication round. In this single round, neighboring nodes distribute enough maximally entangled states to support all the communication that $\algA$ would perform during its $T$ rounds. Afterwards, every message that $\algA$ would send is simulated by teleportation, but without sending the classical teleportation outcome. We call this procedure \emph{wishful teleportation}.

Recall that in standard teleportation \cite{bennett-brassard-etal-1993-teleporting-an-unknown}, a node $p$ can transmit a quantum state $\ket{\psi}$ to a node $p^\prime$ using a shared maximally entangled state $\ket{\Phi}$ and classical communication.
To do so, $p$ performs a joint measurement on $\ket{\psi}$ and its share of $\ket{\Phi}$, obtaining a classical outcome $k$, also known as the \emph{teleportation key}. The system held by $p^\prime$ then becomes similar to $\ket{\psi}$, up to some correction $P_k$.
Thus, once $p$ communicates the classical key $k$ to $p^\prime$, the exact state $\ket{\psi}$ can be recovered by $p^\prime$.
In wishful teleportation, we replace the communication of $k$ by a guess of translation-invariant keys.
This is likely to fail, but our reduction algorithm $\algB$ can exploit it by outputting the guessed keys together with the outcome of the attempted simulation of $\algA$: if every guess is correct, then the simulation algorithm $\algB$ faithfully reproduces $\algA$, which by assumption breaks symmetry; otherwise, the incorrect guess itself breaks symmetry.
This guarantees that symmetry is broken by a one-round algorithm $\algB$ whether $\algA$ is faithfully simulated or not.

More precisely, in each round of communication $t\in\{1,2, ...,T\}$ simulated by wishful teleportation, node $p$ obtains two teleportation keys: $k_{t,L}(p)$ and $k_{t,R}(p)$ for the states it respectively attempts to teleport to its neighbors on the left (node $p-1$) and right (node $p+1$). To recover the states sent by nodes $p-1$ and $p+1$, node $p$ would need the keys $k_{t,R}(p-1)$ and $k_{t,L}(p+1)$, which are not communicated.
Instead, $p$ simply guesses that they satisfy
\begin{equation}
\label{eq:intro-symmetry-not-broken}
k_{t,R}(p) = k_{t,R}(p-1)
\quad\text{and}\quad
k_{t,L}(p)=k_{t,L}(p+1) ,
\end{equation}
and applies the corrections $P_{k_{t, R}(p)}$ and $P_{k_{t, L}(p)}$ to the appropriate registers.
Whenever these equalities hold, the corresponding teleportation step succeeds.
See \Cref{fig:teleport-keys} for an illustration.

\begin{figure}[ht!]
    \centering
    \includegraphics[height=0.5\linewidth]{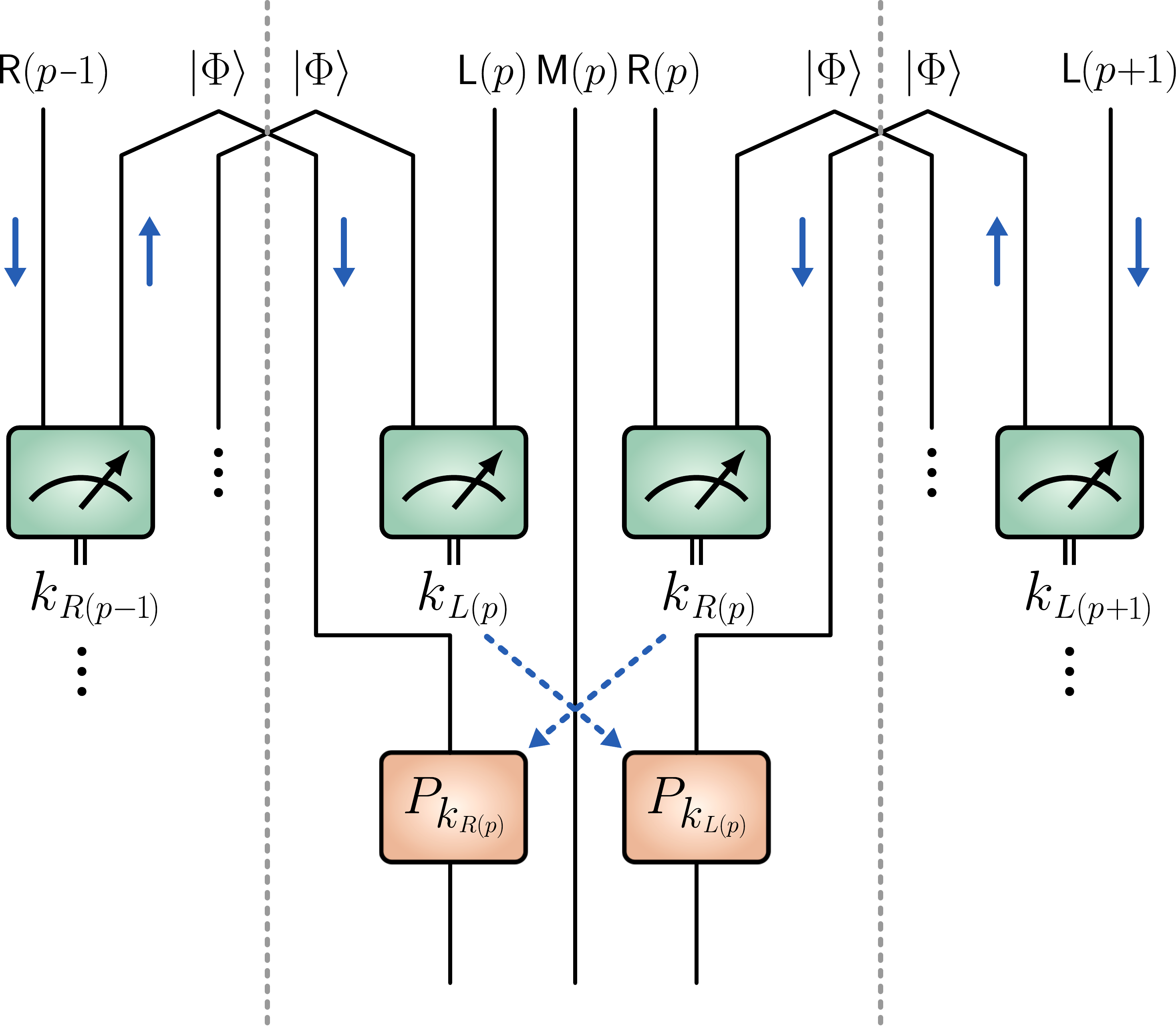}
    \caption{In wishful teleportation, node $p-1$ (resp.\ $p+1$) attempts to teleport the state in $R(p-1)$ (resp.\ $L(p+1)$) to register $L(p)$ (resp.\ $R(p)$) of node $p$. Node $p$ assumes $k_{R(p-1)}=k_{R(p)}$ (resp.\ $k_{L(p+1)}=k_{L(p)}$) and applies the correction $P_{k_{R(p)}}$ (resp.\ $P_{k_{L(p)}}$). Blue arrows are used to indicate the flow of information.}
    \label{fig:teleport-keys}
\end{figure}

Repeating this procedure for all $t \in [T]$ simulated rounds gives each node a list of teleportation keys. After the simulation, each node $p$ performs the final measurement prescribed by $\algA$ and obtains $a(p)$. The output of $p$ in $\algB$ is then
\[
\bigl(k_{1,L}(p),k_{1,R}(p),\ldots,k_{T,L}(p),k_{T,R}(p),a(p)\bigr) .
\]

Now consider the output of any block of $2T+2$ consecutive nodes.
If every teleportation guess inside the block is correct, then $\algA$ is simulated faithfully on its two central nodes, which receive different colors because by assumption $a(p) \neq a(p+1)$ in $\algA$.
Otherwise, some teleportation guess is incorrect.
This means that at least one condition in \Cref{eq:intro-symmetry-not-broken} fails somewhere inside the block, thus at least one pair of guessed teleportation keys differ between two nodes in that block.
Since these keys are included in the output, the outputs on the block cannot all be identical.
Hence, $\algB$ breaks symmetry in a single communication round.

\paragraph{Finitely many qubits and error probability.}
Our analysis of one-round algorithms only rules out the existence of a finite-dimensional quantum algorithm, with local systems of dimension $d^3$, that solves $3$-coloring in $O(1)$ rounds with success probability $1$.
Since the error probability of an algorithm can be written as a polynomial with integer coefficients and real variables, once it is known to be non-zero, results from optimization (see \cite[Theorem 1.1]{jeronimo-perrucci-tsigaridas-2013-on-the-minimum-of-a}) provide a lower bound on the error probability: for each fixed $d$, the success probability is upper bounded by $1-\varepsilon(d)$, for some $\varepsilon(d)>0$ that may go to zero as $d\rightarrow\infty$.
Because our argument for the one-round impossibility (\Cref{sec:one-round}) uses an MPS representation, where the local dimension plays a central role, it is not clear how to obtain a dimension-independent bound $1-\varepsilon$ with $\varepsilon>0$ independent of $d$.

Such a dimension-independent lower bound on the error probability would concern infinite-dimensional quantum strategies as defined by the \emph{tensor-product} model of quantum theory (see e.g. Def.~1 and Def.~1, 4 of \cite{ligthart-gross-2023-the-inflation-hierarchy-and-the}). This model extends finite-dimensional quantum theory by allowing arbitrary Hilbert spaces while preserving the tensor-product rule for combining separated systems.\footnote{In the standard Bell scenario, the corresponding set of correlations is usually denoted $\mathcal{C}_{qs}$. Since this set is not closed in general~\cite{slofstra-2019-the-set-of-quantum-correlations-is-not}, one often considers its closure $\mathcal{C}_{qa}$.}

There is, however, a more general \emph{commuting-operator} model of quantum theory, where independent systems are associated to commuting operator algebras acting on a common Hilbert space (see e.g. Def.~2, 3, 5 of \cite{ligthart-gross-2023-the-inflation-hierarchy-and-the}). This model strictly contains the closure of the tensor-product model, as follows from the breakthrough result $\mathrm{MIP}^*=\mathrm{RE}$~\cite{ji-natarajan-etal-2021-mip-re}.\footnote{Unlike $\mathcal{C}_{qs}$, the set $\mathcal{C}_{qc}$ is closed.}

For standard bipartite Bell scenarios, the main systematic tool for proving that a distribution is impossible in the commuting-operator model is the Navascués--Pironio--Acín (NPA) hierarchy~\cite{navascues-pironio-acin-2007-bounding-the-set-of-quantum,navascues-pironio-acin-2008-a-convergent-hierarchy-of}. In contrast, the result $\mathrm{MIP}^*=\mathrm{RE}$ implies that there can be no analogous complete decision procedure for the tensor-product model in full generality. Thus, a natural route for extending our lower bound beyond finite-dimensional exact strategies is to reformulate it directly in the commuting-operator model, and to use distributed analogues of the NPA hierarchy, such as the inflation-NPA hierarchy~\cite{wolfe-pozas-kerstjens-etal-2021-quantum-inflation-a}.

 \section{Preliminaries}
\label{sec:background}

\paragraph{Mathematical Notation.}
We use $\bN$, $\bR$, and $\bC$ to denote the integers greater than zero, the reals, and the complex numbers respectively. For an integer $n \geq 1$, we use $[n] = \set{1, 2, \ldots, n}$. For a Hilbert space $H$, we use $\ket{u}$ to denote vectors of $H$ and $\bra{u}$ for the corresponding co-vector. The space of complex $d \times d$ matrices is written $\Mat_d(\bC)$ and the space of endomorphisms of a vector space $E$ is denoted by $\End(E)$. The identity matrix is written as $\One$. Given an operator $P$, we write the adjoint (i.e., the conjugate-transpose), the conjugate and the transpose as $P^\dagger$, $\overline{P}$ and $P^\top$, respectively.

\subsection{Quantum Information Theory Basics}
\label{sec:prelim-quantum}
We present here the essential for following our proofs and refer readers to \cite{nielsen-chuang-2010-quantum-computation-and-quantum,watrous-2018-the-theory-of-quantum-information} for an extensive introduction.

For $d\in \bN$, a $d$-dimensional \emph{register} is a complex Euclidian space isomorphic to $\bC^d$ and represents a physical device holding information. Throughout, we use sans-serif letters $\rX, \rY, \rZ$ to represents registers. A \emph{state vector} on $\rX$ is simply a vector $u\in \rX$. A \emph{(quantum) state} on register $\rX \simeq \bC^d$ is a density matrix $\rho$: (1) $\rho$ is a $d \times d$ matrix with complex coefficients, (2) $\rho$ is positive semi-definite, and (3) the trace of $\rho$ is one. To provide some intuition, recall that from the spectral theorem, any positive semi-definite matrix (and, in particular, any density matrix) $\rho$ can be written as 
\begin{equation}  
  \label{eq:decompose-density-matrix}
  \rho = \sum_{i=1}^r \lambda_i \ketbra{u_i}{u_i} \ ,
\end{equation}
where the $\lambda_1, \lambda_2, \ldots, \lambda_r \in \bR_{\geq 0}$ are the eigenvalues of $\rho$ and $\ket{u_1}, \ket{u_2}, \ldots, \ket{u_r} \in \bC^d$ are the corresponding eigenvectors. Since the trace of $\rho$ is one, we have that $\sum_i \lambda_i = 1$; in other words, $\rho$ describes a probability distribution over state vectors. When $\rho = \ketbra{u}{u}$ has rank one, we say that it is a \emph{pure state}. A state on two registers $\rX,\rY$ is a state on $\rX \otimes \rY$.

Operations on a quantum state are in full generality represented by quantum channels. For the purpose of this paper, we only need two kinds of operations on quantum states: applying a unitary, and performing a measurement followed by some classical computation. If $\rX$ is in a state $\rho$ and we apply the unitary $U$ on $X$, the state becomes $U\rho U^\dagger$. For intuition, consider applying $U$ to each $\ket{u_i}$ in \cref{eq:decompose-density-matrix} and observe the resulting density matrix. A \emph{measurement} on $\rX$ with outputs in a \emph{finite} set $\rSigma$ produces a random element in $\rSigma$. Formally, a measurement on $\rX$ is described by a Positive Operator-Valued Measure, henceforth \emph{POVM}: a tuple of positive semi-definite matrices $\Pi_a$ for $a\in \rSigma$ such that $\sum_{a\in\rSigma} \Pi_a = \One$. The output probability distribution is given by \emph{Born's rule}: the probability of output $a\in \rSigma$ is $\tr{ \Pi_a \rho }$. 

\subsection{The \qLOCAL and \qPN models}
\label{sec:quantum-local}

The \qLOCAL model is a model of distributed quantum computing. The network is represented by a graph where each vertex is a quantum processor, and each edge allows quantum communication between neighboring nodes. Computation proceeds in synchronous rounds. In each round, every node may perform arbitrary local quantum computation and exchange quantum messages of unbounded size with its neighbors. After the last round, each node performs a quantum measurement and outputs a classical string. Thus, \qLOCAL is a natural quantum analogue of the classical LOCAL model of distributed algorithms\footnote{In the classical setting, there exist two equivalent descriptions of the LOCAL model: one based on message-passing and one based on mapping local neighborhoods to outputs (see, e.g., \cite[Chapter 4]{hirvonen-suomela-2020-distributed-algorithms-2020}). When nodes of the network in the former description have access to qubits, no formulation akin to the latter is known.}, obtained by replacing local classical computation and communication with their quantum counterparts.

In LOCAL, hence also in \qLOCAL, the nodes are given unique identifiers, usually of $O(\log n)$ bits. When nodes are not provided unique identifiers, the classical model is called the \emph{port-numbering} model, henceforth PN. Randomized PN algorithms that err with probability at most $1/\poly(n)$ can sample random identifiers before the first round, hence randomized PN was heavily used to prove lower bounds in the LOCAL model (e.g., in \cite{brandt-fischer-etal-2016-a-lower-bound-for-the,balliu-brandt-etal-2019-lower-bounds-for-maximal}). In this paper, our focus is the \emph{\qPN} model, i.e., the \qLOCAL model in which nodes are not provided unique identifiers. As the lower bound from \cref{thm:coloring} applies only to zero-error protocols, it does \emph{not} imply a lower bound in the \qLOCAL model.

In this paper, the underlying communication graph is always a \emph{directed cycle}, hence the adjacency relations are implicit in the following definition. Each node has three registers --- left, middle, and right --- to which it can apply a local quantum transformation, and a round of communication swaps the content of the left and right registers between adjacent nodes. To specify an algorithm, it therefore suffices to give a sequence of local quantum transformations and the final measurement.

In full generality, a quantum transformation is described by a Completely Positive Trace Preserving (CPTP) map, see \cite[Eq.\ 1.124 and 1.125]{watrous-2018-the-theory-of-quantum-information} for the definition.
By Stinespring dilation (\cite[Proposition 2.20]{watrous-2018-the-theory-of-quantum-information}), any CPTP map can be implemented by a unitary transformation on a larger Hilbert space.
Since the local dimension in the \qPN model is arbitrary, this enlarged space can be embedded in the model.
Thus, \qPN algorithms can be defined in terms of unitary transformations without loss of generality.
The final measurement is formally described by a POVM as defined in \cref{sec:prelim-quantum}.
See \cref{fig:example-quantum-LOCAL} for an illustration. 

\begin{definition}\label{def:q-local}\label{def:q-PN}
A $T$-round \qPN algorithm with output alphabet $\rSigma$ consists of
\begin{itemize}
    \item unitary matrices $U_0, U_1, U_2, \ldots, U_T$ on $\bC^d \otimes \bC^d \otimes \bC^d$ for some $d\in \bN$, and
    \item a POVM $\Pi = (\Pi_a : a \in \rSigma)$ on the same space.
\end{itemize}
\end{definition}

For a given integer $n \geq 1$, a \qPN algorithm induces a distribution on $\rSigma^n$ representing the joint output of the nodes. The nodes correspond to integers of $[n]$ and if the output is $(x_1, x_2, \ldots, x_n) \in \rSigma^n$, then $x_p$ corresponds to the output of node $p$. The (global) output distribution is obtained by running the following process.
\begin{definition}\label{def:run-qLOCAL}
Let $n\in\bN_{\geq1}$. On the $n$-cycle, each vertex $p \in [n]$ has registers
\begin{equation*}
  \rL(p)\otimes\rM(p)\otimes\rR(p)\cong
  \bC^d\otimes\bC^d\otimes\bC^d ,
\end{equation*}
with indices understood cyclically, i.e., $\rL(n+1) = \rL(1)$ and $\rR(0) = \rR(n)$. Initially, the global state is $\bigotimes_{p = 1}^n U_0 \ket{ 000 }_{\rL(p)\rM(p)\rR(p)}$.
For $t = 1, 2, \ldots, T$, each vertex $p$ does
\begin{enumerate}
    \item\label[step]{step:swap-left} swap the content of $\rL(p)$ with $\rR(p-1)$, and \item\label[step]{step:unitary} apply $U_t$ to $\rL(p) \otimes \rM(p) \otimes \rR(p)$.
\end{enumerate}
Finally, each vertex $p$ measures $\Pi$ on $\rL(p) \otimes \rM(p) \otimes \rR(p)$ and outputs the observed label.
\end{definition}

Each iteration of \cref{step:swap-left,step:unitary} is called a \emph{round}. This model assumes that vertices have a consistent orientation of the $n$-cycle, i.e., they agree on which side is the left one.
An algorithm that uses this information can only be faster than an undirected algorithm, therefore lower bounds within \Cref{def:q-local,def:run-qLOCAL} also hold for undirected models.
 \section{Proof of \cref{thm:coloring}}

In this section, we formally state our two main technical contributions, \cref{thm:reduction,thm:no-one-round-symmetry-breaking}, and use them to prove \cref{thm:coloring}: any $T(n)$-round \qPN algorithm that $c$-colors $n$-cycles with $T(n) \leq n/2-2$ must produce a monochromatic edge with non-zero probability. The choice of the \qPN algorithm can depend on $n$, in other words the nodes have knowledge of $n$.

Let us introduce some terminology first. For a \qPN algorithm $\algA$, the \emph{support} of $\algA$ is the set of vectors $\vec{a} \in \bigcup_{n\geq 3} \rSigma^n$ that $\algA$ outputs with non-zero probability.
For $r \in \bN$, an $r$-\emph{block} is a set of $r$ consecutive natural numbers. We say that the $r$-block $B$ is \emph{monochromatic} in $\vec{a} = (a_1, \ldots, a_n) \in \rSigma^n$ if and only if all $a_i$ for $i\in B$ have the same value. In particular, a coloring algorithm that errs with zero probability is such that all output vectors containing a monochromatic 2-block have probability zero.

First, we prove that any 1-round \qPN algorithm produces a monochromatic $(n-2)$-block with non-zero probability, no matter the number of output labels.

\begin{restatable}{theorem}{ThmNoOneRoundSymBreaking}
    \label{thm:no-one-round-symmetry-breaking}
    Let $n \geq 4$ an integer. Any 1-round \qPN algorithm has a monochromatic $(n-2)$-block in its support.
\end{restatable}

Then, we prove that any $T(n)$-round $c$-coloring algorithm can be turned into a 1-round algorithm where nodes output values from $[d]^{4T} \times [c]$ and produces no monochromatic $(2T(n)+2)$-block.

\begin{restatable}{theorem}{ThmReduction}\label{thm:reduction}
    Let $r \in \bN$ and $n \geq 4$ and $T \leq \lfloor \frac{n - r}{2} \rfloor$.
    Suppose there is a $T$-round \qPN algorithm $\algA = (U_0, U_1, \ldots, U_T, \Pi)$ of local dimension $d \in \bN$ with outputs in $\rSigma$ and no monochromatic $r$-block in its support. Then there exists a 1-round algorithm $\algB$ with outputs in $[d]^{4T} \times \Sigma$ and no monochromatic $(2T+r)$-block in its support.
\end{restatable}

\begin{proof}[Proof of \cref{thm:coloring}]
Fix $n\geq4$, and let $\algA$ be a \qPN algorithm that properly
$c$-colors the directed $n$-cycle in $T(n) \leq n/2 - 1$ rounds, i.e., it has $\rSigma = [c]$ and no monochromatic 2-block in its support.
By \Cref{thm:reduction}, there is a one-round \qPN algorithm $\algB$ with output $[d]^{4T} \times[c]$ with no monochromatic $(2T(n)+2)$-block in its support.
By \Cref{thm:no-one-round-symmetry-breaking}, the algorithm $\algB$ produces a monochromatic $r$-block with non-zero probability, where $r = n-2$ when $n$ is even, and $r = n-3$ when $n$ is odd. In particular, $\algB$ has a monochromatic $r$-block in its support for all $r \leq n-2$. Hence, we must have $2T(n)+2 > n-2$, meaning that $T(n) > \floor{n/2} - 2$.
\end{proof}
 \section{No One-Round Quantum Symmetry Breaking (\cref{thm:no-one-round-symmetry-breaking})}
\label{sec:one-round}

In this section, we prove that symmetry breaking is not possible with probability one in one-round of \qPN. Recall that $n$, the number of nodes in the cycle, is fixed first and the one-round \qPN can depend on it, i.e., nodes know the length of the cycle.

\ThmNoOneRoundSymBreaking*

The proof of \cref{thm:no-one-round-symmetry-breaking} follows from \cref{lem:encoding,lem:block}. The former describes the output distribution as the trace of product of certain matrices, while the latter provides some structure on those matrices. We state the lemmas here, derive the proof of \cref{thm:no-one-round-symmetry-breaking}, and prove them in \cref{sec:proof-lem-encoding} and \cref{sec:proof-block} respectively. 

Call $d\in \bN$ the dimension of the local registers: $\rL(p), \rM(p), \rR(p) \simeq \bC^d$ and let $H = \bC^d \otimes \bC^d \otimes \bC^d$. The communication round swapping each $\rL(p)$ with $\rR(p-1)$ is implemented by the unitary matrix $C$ defined as
    \[
    C \paren*{ \bigotimes_{p = 1}^n \ket{ \ell_p }_{\rL(p)} \ket{ m_p }_{ \rM(p) } \ket{ r_p }_{\rR(p)} }
    = \bigotimes_{p = 1}^n \ket{ r_{p-1} }_{\rL(p)} \otimes \ket{ m_p }_{ \rM(p) } \otimes \ket{ \ell_{p+1} }_{ \rR(p) } \ ,
    \]
    for all $\ell_p, m_p, r_p \in [d]$ and all $p\in \set{0, 1, \ldots, n}$ such that $\ell_0 = \ell_n$ and $r_0 = r_n$. 
In \cref{lem:encoding}, the state $u$ should be thought of as the local state of each player \emph{before} the first (and unique) communication round, and $C\ket{ u, \ldots, u}$ as the global state after the communication round.

\begin{restatable}{lemma}{LemEncoding}
    \label{lem:encoding}
    For all $n \geq 4$, $\ket{u} \in H \setminus \set{ 0 }$,
    there exists a conjugate-linear map\footnote{a conjugate-linear map $f : U \to V$ between complex vector spaces is additive, i.e., $f(x + y)=f(x) + f(y)$ for all vectors $x,y\in U$, but conjugates the scalar factors, i.e., $f(\lambda x) = \overline{\lambda}f(x)$ for all $x\in U$ and $\lambda\in \bC$} $L : H \to \Mat_{d^2}(\bC)$ and a Hilbert space $V \subseteq \Mat_d(\bC)$ with $0 < \dim(V) \leq d$ such that 
    \[
    \im(L) = ( V \otimes \Mat_d(\bC) ) P
    \quad\text{where}\quad
    \forall \alpha,\beta \in \bC^d \text{ we have } P( \alpha \otimes \beta ) = \beta \otimes \alpha \ ,
    \]
    and for all $v_1, \ldots, v_n \in H$, we have that
    \[
    \bra{ v_1, v_2, \ldots, v_n } C \ket{ u, \ldots, u } = 
    \tr*{ \prod_{p=1}^n L(v_p) } \ .
    \]
\end{restatable}

The matrices $M = L(x)$ for some $x\in H$ have some special structure: indeed, with respect to the right basis, the matrix $M$ has the following block decomposition
\[
M = 
\left[
    \begin{array}{cccc}
        \tau(M) & 0 & * & 0 \\
        * & 0 & * & 0 \\
        0 & 0& 0 & 0 \\
        0 & 0& 0 & 0 
    \end{array} 
\right] \ .
\]
We also use that this decomposition is compatible with products of matrices and that the map $\tau$ is surjective.
Formally, we prove the following; see \cref{sec:proof-block} for the proof.

\begin{restatable}{lemma}{LemBlock}
    \label{lem:block}
    Let $d\in \bN \setminus \set{0}$ and $V$ be a subspace of $\Mat_d(\bC)$ with $0 < \dim(V) \leq d$. Call $S = ( V \otimes \Mat_d(\bC) ) P$ and $S^k = \operatorname{span}\set{ M_1 M_2 \ldots M_k : M_1, M_2, \ldots, M_k \in S}$.
    There exists a subspace $E$ of $\bC^d \otimes \bC^d$ and a map $\tau : \Mat_{d^2}(\bC) \to \End(E)$
with the following properties:
    \begin{enumerate}[label=(T\arabic*)]
        \item\label[part]{part:tau-mul}
        $\tau(M N) = \tau(M) \tau(N)$ for all $M, N \in \bigcup_{k \geq 1} S^k$,
\item\label[part]{part:tau-trace}
        $\tr{ M } = \tr{ \tau(M) }$ for all $M \in \bigcup_{k \geq 1} S^k$,
        \item\label[part]{part:tau-surj}
        for all $X \in \End(E)$, there exists $M\in S^2$ such that $\tau(M) = X$. 
    \end{enumerate}
\end{restatable}

\noindent
Given these technical lemmas, we are ready to derive the proof of \cref{thm:no-one-round-symmetry-breaking}.

\begin{proof}[Proof of \Cref{thm:no-one-round-symmetry-breaking}]
    The proof is by contradiction. Suppose there exists an algorithm $\algA = (U_0, U_1, \Pi)$ with output labels $\rSigma = [c]$ for some $c\in \bN$ and no monochromatic $(n-2)$-block in its support. Denote the local state of every vertex before the communication round by $u = U_0 \ket{ 000 }$. In particular, the global state of after the communication round is $\rho = C \ket{u}^{\otimes n} \bra{u}^{\otimes n} C^\dagger$. 
    Let $V$ and $L$ be as described by \cref{lem:encoding} and $\tau$ the map given by \cref{lem:block}. One can absorb the unitary $U_1$ into the measurement by replacing each $\Pi_a$ with $U_1^\dagger \Pi_a U_1$; so we henceforth assume that $U_1$ is the identity.
    
    Since $\algA$ has no monochromatic $(n-2)$-block in its support, the POVM $\Pi = (\Pi_1, \ldots, \Pi_c)$ is such that $\tr{ (\Pi_a^{\otimes n-2} \otimes I \otimes I) \cdot \rho } = 0$ for all $a\in [c]$. 
    Using that the $\Pi_a$ are positive semi-definite matrices, we have that $(\Pi_a^{\otimes n-2} \otimes \One \otimes \One) C\ket{ u, \ldots, u }  = 0$ for all $a\in [c]$.
    In particular for all $x, y, z\in H$, we have that
    \begin{align}
        \label{eq:zero-trace}
        \begin{split}
    \tr{ L(\Pi_a^\dagger x)^{n-2} L(y) L(z) }
    &= \bra{ x\Pi_a, \ldots, x\Pi_a,  y, z } C \ket{ u, \ldots, u }  \\
    &= \bra{ x, \ldots x, y, z } ( \Pi_a^{\otimes n-2} \otimes \One \otimes \One) C \ket{ u, \ldots, u }  = 0 \ .
        \end{split}
    \end{align}

    Let us now deduce that $\tr{ L(x) } = 0$ for all $x\in H$.
    Fix a color $a\in[c]$ and $x\in H$.
    If $\tau(L(\Pi_a^\dagger x))^{n-2} \neq 0$, there exists a matrix $X$ for which $\tr{ \tau(L(\Pi_a^\dagger x))^{n-2} X } \neq 0$. By \ref{part:tau-surj}, there exists $M \in S^2$ such that $\tau(M) = X$. Therefore, 
    \begin{align*}
        \tr{ L(\Pi_a^\dagger x)^{n-2} M } 
        &\stackrel{\text{\ref{part:tau-trace}}}{=} \tr{ \tau( L(\Pi_a^\dagger x)^{n-2} M ) } \\
        &\stackrel{\text{\ref{part:tau-mul}}}{=} \tr{ \tau( L(\Pi_a^\dagger x) )^{n-2} \tau( M ) }  \\
        &= \tr{ \tau( L(\Pi_a^\dagger x) )^{n-2} X } \neq 0 \ . 
    \end{align*}
    Rewrite $M = \sum_i \lambda_i L(y_i) L(z_i)$ with $y_i, z_i\in H$ as $M \in S^2$. By additivity and \cref{eq:zero-trace}, we have 
    \[ 
    \tr{L(\Pi_a^\dagger x)^{n-2} M} 
    = \tr*{ \sum_i \lambda_i L(\Pi_a^\dagger  x)^{n-2} L(y_i) L(z_i) }
    = \sum_i \lambda_i \tr*{ L(\Pi_a^\dagger  x)^{n-2} L(y_i) L(z_i) }
    = 0 \ ,
    \]
    which is absurd. Hence, it must be that $\tau(L(\Pi_a^\dagger  x))^{n-2} = 0$ and, in particular, $\tr{ \tau( L(\Pi_a^\dagger  x) ) } = 0$ because any nilpotent matrix has trace zero\footnote{If $A^k = 0$, then all eigenvalues $\lambda$ of $A$ are such that $\lambda^k = 0$ because $A^k v = \lambda^k v$ for $v$ a corresponding eigenvector. Hence, all eigenvalues of $A$ are zero, hence $\tr{A} = 0$}. By \ref{part:tau-trace} we get that $\tr{ L(\Pi_a^\dagger  x) } = \tr{ \tau(L(\Pi_a^\dagger  x)) } = 0$. 
    It follows that $\tr{L(x)} = 0$ for all $x\in H$ by additivity as
    \[
    \tr{ L(x) } = \tr*{ L\paren*{ \sum_{a\in [c]} \Pi_a^\dagger  x } }
    = \sum_{a\in[c]} \tr*{ L(\Pi_a^\dagger  x) }
    = 0 \ .
    \]

    To conclude the proof, pick any $T \in V \setminus \set{0}$, which exists because $\dim V > 0$ from \cref{lem:encoding}. Since $T \neq 0$, there exists $i, j$ such that $T_{i,j} \neq 0$. Now 
    \[ 
    \tr{(T \otimes E_{j,i})P} 
    = \tr{ T E_{j,i} } = T_{i,j} \neq 0 \ ,
    \]
    which leads to a contradiction because $(T \otimes E_{j,i})P \in \im(L)$ by \cref{lem:encoding} and thus should have trace zero.
\end{proof}

\subsection{Proof of \cref{lem:encoding}}
\label{sec:proof-lem-encoding}

In \Cref{lem:encoding}, the representation of the amplitudes of $C\ket{u,\ldots,u}$ is a Matrix Product State (MPS) representation \cite{fannes-nachtergaele-werner-1992-finitely-correlated,perez-garcia-verstraete-etal-2007-matrix-product-state}, which is extensively studied in condensed matter physics and quantum information theory \cite{cirac-perez-garcia-etal-2021-matrix-product-states-and}.
To the best of our knowledge, this is the first application to distributed quantum algorithms.
However, this representation alone does not suffice for our argument; the crucial point is that, for states generated by one round of communication, the map $L$ can be chosen with the specific image structure described in \Cref{lem:encoding}.
The proof is provided here in full details for completeness.

Let us restate the lemma.
\LemEncoding*

\begin{proof}

For clarity, let $\rL$, $\rM$ and $\rR$ be three $d$ dimensional Hilbert spaces and $H = \rL \otimes \rM \otimes \rR$.
Using the Schmidt decomposition (see, e.g., \cite[Eq. 1.164]{watrous-2018-the-theory-of-quantum-information}) of $\ket{u}$ w.r.t.\ $\rM$ and $\rL \otimes \rR$, we get an \emph{orthonormal} family $\set{m_1, \ldots, m_r} \subseteq \rM$ and an \emph{orthogonal} $\set{T_1, \ldots, T_r} \subseteq \rL \otimes \rR$ for which 
\[
\ket{ u } = \sum_{i = 1}^r \ket{ m_i }_{\rM} \otimes \ket{ T_i }_{\rL,\rR}  \ .
\]
Define $V = \operatorname{span}\set{ T_1, \ldots, T_r }$, which we see as a subspace of $\Mat_d(\bC)$ of dimension at most $r \leq d$. For each $i \in [r]$ and vector $\ket{v}\in H$, define $B_i(v) \in \Mat_d(\bC)$ as $(B_i(v))_{a,b} = \braket{ v }{a, i, b}$, where $m_i = \ket{i}_{\rM}$ for succintness. Now we can define the map $L$ as 
\[
L(v) =  \paren*{ \sum_{i = 1}^r T_i \otimes B_i(v) } P \ ,
\]
which is clearly conjugate-linear and $\im(L) \subseteq (V \otimes \Mat_d(\bC)) P$. Conversely, the orthonormality of the $m_i$ gives that $B_i(v) = E_{j,k}$ when $v = \ket{j, i, k}$. Since the $T_i$ span $V$ and the $E_{j,k}$ span $\Mat_d(\bC)$, we get that $\im(L) = (V \otimes \Mat_d(\bC)) P$. So what remains to prove is the trace property.

First, rewrite $\ket{u} = \sum_i \ket{ m_i }_{\rM} \otimes \ket{ T_i }_{\rL,\rR}$ as an element of $\rL \otimes \rM \otimes \rR$ as follows
\[
\ket{u} = \sum_{i = 1}^r \sum_{j, k \in [d]} (T_i)_{j,k}  \ket{j, i, k}_{\rL, \rM, \rR}\ .
\]
Therefore, if we have $p$ copies of this state $\ket{u}$ in registers $\rL(p), \rM(p),  \rR(p)$ (using cyclic indices $k_0 = k_n$ and $j_{n+1} = j_1$), after the communication round the global state can be written as
\begin{align*}
    C\ket{u, \ldots, u}
    &= C\paren*{ \sum_{i_1, \ldots, i_n\in [r]} \sum_{j_1, \ldots, j_n \in [d]} \sum_{k_1, \ldots, k_n\in [d]} \paren*{ \prod_{p = 1}^n (T_{i_p})_{j_p,k_p} } \bigotimes_{p=1}^n \ket{ j_p, i_p, k_p }_{\rL(p), \rM(p), \rR(p)} } \\
    &= \sum_{i_1, \ldots, i_n\in [r]} \sum_{j_1, \ldots, j_n \in [d]} \sum_{k_1, \ldots, k_n\in [d]} \paren*{ \prod_{p = 1}^n (T_{i_p})_{j_p,k_p} } \bigotimes_{p=1}^n \ket{ k_{p-1}, i_p, j_{p+1} }_{\rL(p), \rM(p), \rR(p)}\\
\end{align*}
Consider any vectors $\ket{ v_1 }, \ldots, \ket{ v_n } \in H$ and note that $\bra{ v_1, v_2, \ldots, v_n} C \ket{ u, u, \ldots, u }$ can be written as
\begin{align*}
\sum_{i_1, \ldots, i_n\in [r]} &\sum_{j_1, \ldots, j_n \in [d]} \sum_{k_1, \ldots, k_n\in [d]} \paren*{ \prod_{p = 1}^n (T_{i_p})_{j_p,k_p} } \bigotimes_{p=1}^n \braket{ v_p }{k_{p-1}, i_p, j_{p+1}}_{\rL(p),\rM(p),\rR(p)} \\
&= \sum_{i_1, \ldots, i_n\in [r]} \sum_{j_1, \ldots, j_n \in [d]} \sum_{k_1, \ldots, k_n\in [d]} \prod_{p = 1}^n (T_{i_p})_{j_p,k_p}  (B_{i_p}(v_p))_{k_{p-1}, j_{p+1}}
\end{align*}
Let us assume for now that $n$ is even; the case of $n$ odd is almost identical and described later.
The key observation to understand this state is that (for $n$ even) the product can be rearranged as
\begin{align*}
\underbrace{\paren*{ (T_{i_1})_{j_1,k_1} (B_{i_2}(v_2))_{k_1,j_3} \ldots 
(T_{i_{n-1}})_{j_{n-1},k_{n-1}} (B_{i_n}(v_n))_{k_{n-1},j_{1}} }}_{\text{all the } j_p, k_p\text{ with odd } p}
\cdot \\
\underbrace{\paren*{ (T_{i_2})_{j_2,k_2} (B_{i_3}(v_3))_{k_2,j_4} \ldots
(T_{i_n})_{j_n,k_n} (B_{i_1}(v_1))_{k_n,j_2} }}_{\text{all the } j_p, k_p\text{ with even } p}
\end{align*}
Therefore $\bra{ v_1, v_2, \ldots, v_n} C \ket{ u, u, \ldots, u }$ factorizes as
\begin{align*}
\sum_{i_1,\ldots,i_n\in[r]} 
&\paren*{ \sum_{j_p, k_p \text{ with $p$ odd}} (T_{i_1})_{j_1,k_1} (B_{i_2}(v_2))_{k_1,j_3} 
\ldots
(T_{i_{n-1}})_{j_{n-1},k_{n-1}} (B_{i_n}(v_n))_{k_{n-1},j_{1}}
} 
\cdot \\
&\paren*{ \sum_{j_p, k_p \text{ with $p$ even}} (T_{i_2})_{j_2,k_2} (B_{i_3}(v_3))_{k_2,j_4} 
\ldots
(T_{i_n})_{j_n,k_n} (B_{i_1}(v_1))_{k_n,j_2} }
\end{align*}
Since each term of the product is a trace, $\bra{ v_1, \ldots, v_n} C \ket{ u, \ldots, u }$ is
\begin{align*}
\sum_{i_1,\ldots,i_n\in[r]} 
\tr{ T_{i_1} B_{i_2}(v_2) \ldots T_{i_{n-1}} B_{i_n}(v_n) }
\tr{ B_{i_1}(v_1) T_{i_2} \ldots B_{i_{n-1}}(v_{n-1}) T_{i_n} }
\end{align*}
Using that $(A \otimes B)P (C \otimes D)P = AD \otimes BC$, one can prove by induction that
\begin{align*}
\tr*{ \prod_{p=1}^n L(v_p) }
&= \sum_{i_1,\ldots,i_n\in[r]} \tr*{ \prod_{p=1}^n (T_{i_p} \otimes B_{i_p}(v_p))P  }\\
&= \sum_{i_1,\ldots,i_n\in[r]} \tr*{  T_{i_1} B_{i_2}(v_2) \ldots T_{i_{n-1}} B_{i_n}(v_n) 
\otimes 
B_{i_1}(v_1) T_{i_2} \ldots B_{i_{n-1}}(v_{n-1}) T_{i_n} } \\
&= \bra{ v_1, v_2, \ldots, v_n} C \ket{ u, u, \ldots, u } \ ,
\end{align*}
which concludes the proof of \cref{lem:encoding} for even values of $n$. For $n$ odd, the analysis is almost identical except that one obtains the following single trace
\[
\sum_{i_1, \ldots, i_n\in [r]} \tr{ T_{i_1} B_{i_2}(v_2) \ldots T_{i_{n}} B_{i_1}(v_1) \cdot T_{i_2} B_{i_3}(v_3) \ldots B_{i_n}(v_n) } \ .
\]
One can then prove by induction that $\tr*{ \prod_{p=1}^n L(v_p) }$ has this form when $n$ is odd, thereby proving \cref{lem:encoding} for odd values of $n$.
\end{proof}

\subsection{Proof of \cref{lem:block}}
\label{sec:proof-block}
Now we prove that the image of $L$ has the claimed block structure.
For intuition, suppose that $V$ is a matrix space of dimension one, spanned by $A$. For any $M\in\Mat_{d}(\bC)$ and $\ket{x,y} \in \bC^{d} \otimes \bC^{d}$, we have that $(A \otimes M)P\ket{x,y} = A\ket{y} \otimes M\ket{x}$. In particular, when $y\in \ker(A)$ this is always zero, and the first factor of the tensor always belong to $\im(A)$. \cref{lem:block} extends this reasoning to the general case by using the sum of the $\im(A)$ and the intersection of the $\ker(A)$.

\LemBlock*

\begin{proof}
Define 
\[
U = \sum_{A \in V} \im(A) \ ,\quad
K = \bigcap_{A \in V} \ker(A)
\quad\text{and}\quad
W = K^\bot \ .
\]
Let $p$ and $q$ be the orthogonal projections on $U$ and $W$, respectively.
The first part is to prove that
\begin{align}
    \label{eq:spaces}
\begin{split}
    V \Mat_d(\bC) &= \set{ T \in \Mat_d(\bC): \im(T) \subseteq U } = p \Mat_d(\bC)\\
    \Mat_d(\bC) V &= \set{ T \in \Mat_d(\bC): T_{|W^\bot} = 0 } = \Mat_d(\bC) q \ .
\end{split}
\end{align}

To prove one inclusion of the first equality, consider $A\in V$ and $M\in\Mat_d(\bC)$. Clearly, $\im(AM) \subseteq \im(A) \subseteq U$. Hence, for all $\ket{ x }\in \bC^d$, $AM\ket{ x } \in U$. Since $p$ is a projection on $U$, it follows that $AM\ket{x} = pAM \ket{x}$, meaning that $AM = p AM$. This proves that $V\Mat_d(\bC) \subseteq p\Mat_d(\bC)$.

To prove the other inclusion, fix an orthonormal basis $\ket{ u_1 }, \ldots, \ket{ u_r }$ of $U$ and recall that $p = \sum_i \ketbra{ u_i }{ u_i }$. By definition, for each $\ket{ u_i }$ there exists matrices $A_{i,1},\ldots, A_{i,m_i} \in V$ and vectors $\ket{ x_{i,1} }, \ldots, \ket{ x_{i,m_i} } \in \bC^d$ such that $\ket{ u_i } = \sum_j A_{i,j} \ket{ x_{i,j} }$. For a matrix $M \in \Mat_d(\bC)$, observe that $\ketbra{ u_i }{ u_i } M = \sum_j A_{i,j} \ketbra{ x_{i,j} }{ u_i } M$ and $\ketbra{ x_{i,j} }{ u_i } M \in \Mat_d(\bC)$. Hence, $\ketbra{ u_i }{ u_i } M \in V \Mat_d(\bC)$, and $pM \in V\Mat_d(\bC)$ follows as the sum of such terms.

To prove the second equality, let $A \in V$ and $M\in \Mat_d(\bC)$. Clearly, $MA\ket{x} = 0$ for all $\ket{x}\in K$. By rewriting $\ket{x} = q\ket{x} + (\One - q)\ket{x}$, this means $MA \ket{x} = MAq\ket{x}$ because $(\One-q)\ket{x} \in K$. Hence, $\Mat_d(\bC) V \subseteq \Mat_d(\bC)q$, proving one direction of the equality.

To prove the other inclusion, let us first observe that $W = R := \sum_{A \in V} \im(A^\dagger)$. To see why, observe that $\ket{x}\in R^\bot$ if and only if $\langle x, A^\dagger y \rangle = 0$ for all $A\in V$ and $\ket{y}\in \bC^d$, which we can rewrite as $\langle Ax, y \rangle = 0$ for all $A\in V$ and $\ket{y}\in \bC^d$. Hence, $\ket{x} \in R^\bot$ is equivalent to $A\ket{x} = 0$ for all $A\in V$, which is equivalent to $\ket{x}\in \bigcap_{A \in V}\ker(A) = K$. Now, let $\ket{w_1}, \ldots, \ket{w_r}$ be an orthonormal basis of $W$. For all $\ket{w_i}$, there exists $A_{i,1}, \ldots, A_{i,m_i} \in V$ and vectors $\ket{x_{i,1}}, \ldots, \ket{x_{i,m_i}} \in \bC^d$ such that $\ket{w_i} = \sum_j A_{i,j}^\dagger \ket{x_{i,j}}$, whose adjoint is $\bra{w_i} = \sum_j \bra{x_{i,j}} A_{i,j}$. For a matrix $M \in \Mat_d(\bC)$, then $M \ketbra{ w_i }{ w_i } = \sum_j M \ketbra{ w_i }{ x_{i,j} } A_{i,j}$. Since $M \ketbra{ w_i }{ x_{i,j} } \in \Mat_d(\bC)$, we get that $M \ketbra{ w_i }{ w_i } \in \Mat_d(\bC)V$, and $Mq \in \Mat_d(\bC)V$ follows as the sum of such terms.

The key observation is that, with respect to the decomposition
\[
\bC^d \otimes \bC^d = (U \oplus U^\bot) \otimes (W \oplus W^\bot) = (U \otimes W) \oplus (U \otimes W^\bot) \oplus (U^\bot \otimes W) \oplus (U^\bot \otimes W^\bot) \ ,
\]
every $M \in \mathcal{A} := \bigcup_{k \geq 1} S^k$ has block form
\[
M = \left[
        \begin{array}{cccc}
            \tau(M) & 0 & * & 0 \\
            * & 0 & * & 0 \\
            0 & 0& 0 & 0 \\
            0 & 0& 0 & 0 
        \end{array} 
    \right]
    \quad\text{where}\quad
    \tau(M) = P_E M_{|E}
    \quad\text{and}\quad
    E = U \otimes W \ .
\]
Before we prove it, let us argue that \cref{lem:block} follows directly from this. Indeed, \ref{part:tau-mul} and \ref{part:tau-trace} are direct from the block decomposition. To prove \ref{part:tau-surj}, fix $X \in \End(U)$ and $Y \in \End(W)$ arbitrarily and define $\widetilde{X}$ and $\widetilde{Y}$ that extends them on $\bC^d$ with zero on $U^\bot$ and $W^\bot$ respectively. Since $\im(\widetilde{X}) \subseteq U$ and $\widetilde{Y}_{|K} = 0$, by \cref{eq:spaces}, we have that $\widetilde{X} \in V\Mat_d(\bC)$ and $\widetilde{Y} \in \Mat_d(\bC)V$. This means that $\widetilde{X} \in V\Mat_d(\bC)$ and $\widetilde{Y} \in \Mat_d(\bC) V$, and thus that $\widetilde{X} \otimes \widetilde{Y} \in V\Mat_d(\bC) \otimes \Mat_d(\bC)V$. This space is equal to $S^2$. Indeed, $S^2$ is spanned by elements of the form $(A\otimes B)P (C \otimes D)P = AD \otimes BC$, where $A,C\in V$ and $B,D\in \Mat_d(\bC)$, thus $S^2 \subseteq V\Mat_d(\bC) \otimes \Mat_d(\bC)V$. Conversely, any simple tensor of the form $AD \otimes BC$ is an element of $S^2$, which implies that $V\Mat_d(\bC) \otimes \Mat_d(\bC)V \subseteq S^2$ by linearity. By construction, on the subspace $E=U\otimes W$, the map $\widetilde{X} \otimes \widetilde{Y}$ acts as $X \otimes Y$, hence $\tau(\widetilde{X}\otimes \widetilde{Y}) = X \otimes Y$. In finite dimension, $T \in \End(U \otimes W)$ can be written as the sum of $X_i \otimes Y_i$ with $X_i \in \End(U)$ and $Y_i \in \End(W)$, hence \ref{part:tau-surj} follows from linearity.

To conclude, it remains to argue about matrices of $\mathcal{A}$ having the claimed block form. It suffices to prove it for matrices $M \in S$ since products of matrices of this form remain of such form. Given the definition of $M$, there exists matrices $A \in V$ and $B\in \Mat_d(\bC)$ for which $M = (A \otimes B) P$. For a vector $\ket{x} \otimes \ket{y} \in \bC^d \otimes \bC^d$, we have $M(\ket{x} \otimes \ket{y}) = A\ket{y} \otimes B\ket{x} \in \im(A) \otimes \bC^d \subseteq U \otimes (W \oplus W^\bot)$. This implies that the last two rows of the block matrix of $M$ are zeros. On the other hand, $A\ket{y} = 0$ when $\ket{y}\in W^\bot$, so $M(\ket{x} \otimes \ket{y}) = 0$ when $\ket{x} \otimes \ket{y}\in \bC^d \otimes W^\bot = (U \oplus U^\bot) \otimes W^\bot$, which implies that the second and fourth column of the block matrix of $M$ are only zeros. This concludes the proof of \cref{lem:block}.
\end{proof}
 \section{Round Reduction (\cref{thm:reduction})}
\label{sec:round-reduction}

In this section, we prove our second main technical contribution.
\ThmReduction*

The mechanism behind algorithm $\algB$ is that each vertex locally instantiates $T$ maximally entangled states with local dimension $d$ for each of its two neighbors, and uses the unique communication round to send the halves of each of those entangled states to the respective neighbor. Each pair of neighbouring vertices thus shares $2T$ maximally entangled states.
Using these states, nodes can attempt to simulate the original $T$-round algorithm $\algA$ through wishful teleportation (see \Cref{fig:teleport-keys}).
At the end of $\algB$, each vertex outputs the outcome of the attempted simulation of $\algA$ together with all generated keys used for wishful teleportation.

\subsection{Background on Teleportation}
\label{sec:teleportation}

Suppose we have a state $\ket{ \psi }_{\rE\rX}$, where $\rE$ is an environment register and $\rX$ is a register held by some node. We wish to turn this state into $\ket{ \psi }_{\rE\rZ}$ where $\rZ$ is held by a different node but without using quantum communication.
This can be achieved with quantum teleportation \cite{bennett-brassard-etal-1993-teleporting-an-unknown}.
Fix the standard maximally entangled vector $\ket{\Phi} = d^{-1/2}\sum_{i=1}^d \ket{ii}$ in $\bC^d\otimes\bC^d$.
Let $\{P_k:k\in[d^2]\}$ be an orthonormal unitary basis and $\ket{\Phi_k}=(\One \otimes P_k^\top)\ket{\Phi}$ the associated maximally entangled basis.
With these conventions, the teleportation identity is
\begin{equation}\label{eq:teleport}
\ket{ \psi }_{\rE\rX} \otimes \ket{ \Phi }_{\rY\rZ}
    = \frac{1}{d} \sum_{k \in [d^2]} \ket{ \Phi_{k} }_{\rX\rY} \otimes (\Id_{\rE} \otimes P_k^\dagger) \ket{ \psi }_{\rE\rZ} \ .
\end{equation}
Suppose registers $\rX\rY$ are measured in the maximally entangled basis, yielding an outcome $k \in [d^2]$.
Then, by applying the correction $P_k$ on register $\rZ$, we recover the state $\ket{ \psi }_{\rE\rZ}$, provided $k$ is known to the party holding $\rZ$.
The explicit environment will be important in our proof, and we emphasize that when the state is teleported, correlations between $\rX$ and $\rE$ are preserved between $\rZ$ and $\rE$.
This is detailed further in \cref{app:teleportation}.

Now consider the case of a chain of vertices, each holding local registers $\rX(p)$, $\rY(p)$, $\rZ(p)$ for $p \in [n]$.
Assume each pair of registers $\rY(p)$ and $\rZ(p+1)$ hold a $\ket{\Phi}$ state, and that the vertices simultaneously teleport the state in their registers $\rX(p)$ to registers $\rZ(p+1)$, for $p \in [n - 1]$.
The next lemma states that the overall effect is a relabelling of the $\rX(1)\ldots\rX(n-1)$ registers to $\rZ(2)\ldots\rZ(n)$.
See \cref{app:chain-teleportation} for a proof.

\begin{restatable}{lemma}{LemChainTeleport}\label{lem:chain-teleport}
Consider $n \geq 2$ vertices in a chain, with each vertex $p \in [n]$ holding three $d$-dimensional registers $\rX(p)$, $\rY(p)$, $\rZ(p)$.
For each $p \in [n-1]$, the registers $\rY(p)$ and $\rZ(p+1)$ share a maximally entangled state $\ket{\Phi}_{\rY(p) \rZ(p+1)}$.

Let $\rE$ be a finite-dimensional environment and let $\rho$ be a state on $\rX(1)\rX(2)\cdots\rX(n-1)\rE$. Suppose that each vertex $p \in [n - 1]$ measures $\rX(p) \rY(p)$ in the maximally entangled basis, obtains outcome $k_p$, and the corresponding correction $P_{k_p}$ is applied on $\rZ(p + 1)$.
Then the resulting state on $\rZ(2)\rZ(3)\cdots\rZ(n)\rE$ is the reduced state obtained from $\rho$ by relabelling each $\rX(p)$ as $\rZ(p+1)$, for $p \in [n - 1]$.
\end{restatable}
 
\subsection{The Simulation Algorithm}
We refer to $\algB$ from \cref{thm:reduction} as the simulation algorithm and to $\algA$ as the original algorithm.
Let us now describe the simulation algorithm formally. In \cref{item:sim-prep}, nodes prepare maximally entangled states in registers $\rT$ (indexed by the round and direction in which they are used), and they share half of each maximally entangled state with their neighbors during \cref{item:sim-com}. Note that \cref{item:sim-com} is the unique communication round of this algorithm. \cref{item:wishful-teleportation-in-sim-algorithm} implements $T$ rounds of wishful teleportation (see \Cref{fig:teleport-keys} for an illustration).
\cref{item:sim-measurement,item:sim-output} perform the final measurement and produce the final output of algorithm $\algB$.

{
\renewcommand{\themdalg}{$\algB$}
\begin{Algorithm}
    \label{alg:simulator}
    One-Round Simulation of a $T$-round algorithm $\algA = (U_0, U_1, \ldots, U_T, \Pi)$

    \medskip
    \begin{enumerate}[(I)]
        \item \label{item:sim-prep}
        Each vertex $p$ prepares registers $\rL_1(p) \otimes \rM(p) \otimes \rR_1(p)$ in the state $U_0 \ket{ 000 }$, registers $\rL_{t}(p)$, $\rR_{t}(p)$ for $t \in [2, T+1]$ in the state $\ket{ 0 }$, and the registers
        \[
        \rT_{t,L}(p) \otimes \rT_{t,L}'(p)
        \quad\text{and}\quad
        \rT_{t,R}(p) \otimes \rT_{t,R}'(p)
        \]
        in the maximally entangled state $\ket{ \Phi }$ for all $t \in [T]$.
        \item \label{item:sim-com}
        Using \emph{one communication round}, for all $p$ and $t\in[T]$ swap
        \[
        \rT_{t,L}'(p) \;\leftrightarrow\; \rR_{t+1}(p-1)
        \quad\text{and}\quad
        \rT_{t,R}'(p) \leftrightarrow \rL_{t+1}(p+1) .
        \]
        \item \label{item:wishful-teleportation-in-sim-algorithm} For $t \in [T]$ times and \emph{with no communication}, each vertex $p$:
        \begin{enumerate}[(i)]
            \label[step]{step:sim-round}
            \item \label[step]{step:sim-teleport-L}
            measures $\rL_{t}(p) \otimes \rT_{t,L}(p)$ in the basis $\set{ \Phi_{k}, k \in [d^2] }$, obtains the outcome $k_{t,L}(p) \in [d^2]$ and then applies the unitary $P_{k_{t,L}(p)}$ on $\rR_{t+1}(p)$;
            \item \label[step]{step:sim-teleport-R}
            measures $\rR_t(p) \otimes \rT_{t,R}(p)$ in the basis $\set{ \Phi_{k}, k \in [d^2] }$, obtains the outcome $k_{t,R}(p) \in [d^2]$ and then applies the unitary $P_{k_{t,R}(p)}$ on $\rL_{t+1}(p)$;
            \item \label[step]{step:unitary-B} applies the unitary $U_t$ on $\rL_{t+1}(p) \otimes \rM(p) \otimes \rR_{t+1}(p)$.
        \end{enumerate}
        \item \label{item:sim-measurement}
        Measure $\Pi$ on $\rL_{T+1}(p) \otimes \rM(p) \otimes \rR_{T+1}(p)$ and obtains $a(p) \in \Sigma$ as the outcome.
        \item\label{item:sim-output}
        Vertex $p$ outputs
        \begin{equation}\label{eq:simulation-output}\lambda_p = ( k_{1,L}(p), k_{1,R}(p), \ldots, k_{T,L}(p), k_{T, R}(p), a(p)) .\end{equation}
    \end{enumerate}
\end{Algorithm}
}

Compared to the \Cref{def:q-local} of the \qPN model, \Cref{alg:simulator} uses mid-circuit measurements and unitaries with classical control.
The well-known \emph{principle of deferred measurement} states that measurements can always be moved to the end of a quantum circuit, and that classically controlled operations can be substituted by conditional quantum gates \cite[Sec.\ 4.4]{nielsen-chuang-2010-quantum-computation-and-quantum}.
Therefore, $\algB$ is still a one-round \qPN algorithm on directed cycles, in the same sense as in \Cref{def:q-PN}.

As proven formally below, the simulation algorithm can be understood as having two behaviours. First, suppose the keys obtained in \Cref{item:wishful-teleportation-in-sim-algorithm} are such that
\begin{equation}
    \label{eq:assumption-eq}
    k_{t,L}(p) = k_{t,L}(p + 1)
    \quad\text{and}\quad
    k_{t,R}(p) = k_{t,R}(p - 1)
\end{equation}
for every $t \in [T]$.
Then, the corrections $P_{k_{t, L}(p)}$ and $P_{k_{t, R}(p)}$ applied by vertex $p$ to the registers $\rR_{t+1}(p)$ and $\rL_{t+1}(p)$ lead to successful teleportations.
Suppose this holds for all adjacent vertices inside an interval $J_s = [s, s + 2T + r - 1]$, where $s \in [n]$ and the interval is understood cyclically.
Then $\algB$ correctly simulates $\algA$ on the interval $[s + T, s + T + r - 1]$, which is at the center of $J_s$ and has length $r$.
Since $\algA$ by assumption has no monochromatic $r$-block, the $a(p)$ part of the outputs of $\algB$ on this central interval cannot be monochromatic.
The second behaviour happens if, for some $t \in [T]$, \Cref{eq:assumption-eq} fails for some pair of vertices inside $J_s$.
Then two adjacent vertices in $J_s$ have different wishful teleportation keys.
Since these keys are included in the output of $\algB$, the interval $J_s$, which has length $2T + r$, cannot be monochromatic.

\subsection{Proof of \cref{thm:reduction}}
Let $s \in [n]$ and recall that $\algA$ is assumed to have no monochromatic $r$-block in its support.
Fix an interval
\[
J_s = [s,s+2T+r-1]
\]
of $2T+r$ consecutive vertices where $T \leq \lfloor \frac{n - r}{2} \rfloor$, hence the interval does not overlap with itself. 
We prove that $J_s$ is monochromatic at the end of $\algB$ with probability zero: 
for all output labels
\[
\lambda=( \underbrace{\kappa_{1,L},\kappa_{1,R},\ldots, \kappa_{T,L},\kappa_{T,R}}_{\kappa},a)\in([d^2])^{2T}\times\Sigma 
\]
the probability that $\lambda_p = \lambda$ for every $p \in J_s$ is zero (recall that $\lambda_p$ is the output of nodes $p$ in $\algB$, see \cref{eq:simulation-output}).
A union bound on $s$ then implies that none of the $J_s$ are monochromatic w.p.\ one, hence that $\algB$ has no monochromatic $(2T+ r)$-block w.p.\ one.

For $0\leq t\leq T$, define
\[
I_s{(t)}=[s+t,s+2T+r-1-t].
\]
At each round, $I_s{(t)}$ shrinks by removing one boundary vertex from each side.
In particular, $I_s(0) = J_s$, and $I_s(T)$ has length $r$.

Let $C_{\kappa,t}(J_s)$ be the event that, for all $1 \leq \tau \leq t$ and all $p \in I_s(\tau)$,
\begin{equation}
    k_{\tau,L}(p) = k_{\tau, L}(p + 1) = \kappa_{\tau, L}
    \quad\text{and}\quad
    k_{\tau,R}(p) = k_{\tau,R}(p - 1) = \kappa_{\tau, R}\,.
\label{eq:def-event-equal-keys}
\end{equation}
Equivalently, in each of the first $t$ rounds, the keys needed to simulate the teleportations into all vertices $p \in I_s(\tau)$ are equal to the corresponding entries of $\kappa$.
We also set $C_{\kappa, 0}(J_s)$ to be the trivial event.

Suppose that the outputs of the nodes in $J_s$ are monochromatic, that is, $\lambda_p = \lambda = (\kappa, a)$ for every $p \in J_s$.
Then each teleportation key in the output of each $p \in J_s$ is equal to the corresponding component of $\kappa$.
Moreover, for every $1 \leq \tau \leq T$ and every $p \in I_s(\tau)$, we have $p-1, p, p+1 \in I_s(\tau - 1) \subseteq J_s$, and thus \Cref{eq:def-event-equal-keys} is satisfied.
Therefore, a monochromatic $J_s$ implies that $C_{\kappa, T}(J_s)$ occurs for some value of $\kappa$.

To rule out this monochromatic output in $J_s$, we will show that after conditioning on $C_{\kappa,T}(J_s)$, the state of $\algB$ in $I_s(T)$ is the same as the state of $\algA$ in this interval.
Then, since by assumption $\algA$ outputs no monochromatic $r$-block, it will follow that $\algB$ cannot have the monochromatic $2T + r$ block $J_s$.

Let $\rho^{(\algA)}_t(I)$ be the reduced state of $\algA$ (see \Cref{def:q-local}) on an interval $I$ after $t$ rounds of $\algA$, which is obtained by tracing out all registers except
\[
\bigotimes_{p\in I}
\rL(p)\otimes\rM(p)\otimes\rR(p) .
\]
Similarly, let $\rho^{(\algB)}_t(I \mid C_{\kappa,t}(J_s))$ be the reduced state of $\algB$ after $t$ simulated rounds, conditioned on $C_{\kappa,t}(J_s)$, on the registers
\[
\bigotimes_{p\in I}
\rL_{t+1}(p)\otimes\rM(p)\otimes\rR_{t+1}(p)\,.
\]
We claim that for every $0 \leq t\leq T$,
\begin{equation}\label{eq:state-induction}
\rho^{(\algB)}_t(I_s(t) \mid C_{\kappa,t}(J_s)) = \rho^{(\algA)}_t(I_s(t))\,.
\end{equation}

The case $t = 0$ holds because $C_{\kappa, 0}(J_s)$ is the trivial event and after tracing out the teleportation registers of $\algB$, it follows from the definition of \Cref{alg:simulator} that both $\algA$ and $\algB$ have the same initial state on the registers being compared.

Then assume \Cref{eq:state-induction} holds for some $t < T$.
We will prove it also holds for round $t + 1$.
Let $I_s(t) = [u,v]$ and $I_s(t + 1) = [u + 1, v - 1]$.
We first analyse the $\rR$ registers, which in $\algA$ are communicated one step clockwise.
In $\algB$, the corresponding transfers happening into the interval $I_s(t + 1)$ are
\[
\rR_{t+1}(p) \longrightarrow \rL_{t+2}(p + 1), \quad p \in [u, v-2]\,.
\]
Let us apply \Cref{lem:chain-teleport} to the interval $[u, v - 1]$.
We identify the registers $\rX, \rY, \rZ$ in \Cref{lem:chain-teleport} with $\algB$'s registers as
\[
    \rX(p) = \rR_{t+1}(p), \quad \rY(p) = \rT_{t+1, R}(p), \quad \rZ(p) = \rL_{t + 2}(p)\,.
\]
Any registers of $\algB$ not treated explicitly are assumed to be included in the environment $\rE$ of \Cref{lem:chain-teleport}.
In $\algB$, after the communication round, the register $\rT_{t + 1, R}(p) = \rY(p)$ is maximally entangled with $\rL_{t + 2}(p + 1) = \rZ(p + 1)$, as required by the lemma for $p \in [u, v - 2]$.
On the event $C_{\kappa, t + 1}(J_s)$, we have
\[
    k_{t+1, R}(p + 1) = k_{t + 1, R}(p) = \kappa_{t + 1, R}, \quad \forall p \in [u, v - 2]\,.
\]
Thus, the unitary $P_{k_{t + 1, R}(p + 1)}$ applied on $\rL_{t+2}(p + 1)$ is the proper teleportation correction operation.
Therefore, by \Cref{lem:chain-teleport}, the overall effect is a relabelling $\rR_{t+1}(p) \longrightarrow \rL_{t + 2}(p + 1)$ for $p \in [u, v - 2]$. 
Since all registers not mentioned above were included in $\rE$, the lemma also guarantees this relabelling preserves any correlations with $\rE$.

The counter-clockwise transfers of the $\rL$ registers are treated analogously, by applying \Cref{lem:chain-teleport} to the interval $[u + 1, v]$ with the opposite ordering.
This gives the relabellings $\rL_{t + 1}(p) \longrightarrow \rR_{t + 2}(p - 1)$ for $p \in [u + 2, v]$.
Together, these two relabellings reproduce the communication step of $\algA$ on $I_s(t+1) = [u+1,v-1]$.
Indeed, for each vertex $p \in I_s(t+1)$, the register $\rL_{t+2}(p)$ comes from $\rR_{t+1}(p-1)$, and the register $\rR_{t+2}(p)$ comes from $\rL_{t+1}(p+1)$, exactly as in one communication round of $\algA$.

Finally, after the communication round $t + 1$, $\algA$ applies $U_{t+1}$ to the three local registers at each vertex.
In $\algB$, this same unitary is applied at \Cref{step:unitary-B} to $\rL_{t+2}(p) \otimes \rM(p) \otimes \rR_{t+2}(p)$, which by the argument above have the same state as the registers of $\algA$ in $I_s(t+1)$.
This shows that, given the event $C_{\kappa,t+1}(J_s)$, 
\[
    \rho^{(\algB)}_{t + 1}(I_s(t + 1) \mid C_{\kappa,t + 1}(J_s)) = \rho^{(\algA)}_{t + 1}(I_s(t + 1))\,.
\]
Therefore, by induction, \Cref{eq:state-induction} holds for $0 \leq t \leq T$.

In particular, \Cref{eq:state-induction} holds for $t = T$, where the interval $I_s(T) = [s + T, s + T + r - 1]$.
Since $I_s(T)$ has length $r$ and $\algA$ has no monochromatic block of length $r$ in its support, then for every $a \in \Sigma$,
\[
    \Trace\bigl(\Pi_a^{\otimes r} \,\rho^{(\algA)}_T(I_s(T)) \bigr) = 0\, ,
\]
where $\Pi_a$ is a measurement effect (cf.~\Cref{def:q-PN}).
Substituting \Cref{eq:state-induction} in the previous equation, 
\[
\Trace\bigl(\Pi_a^{\otimes r} \rho^{(\algB)}_T(I_s(T) \mid C_{\kappa,T}(J_s)) \bigr) = 0\,.
\]
This is the probability that all vertices in $I_s(T)$ output $a$ given that $C_{\kappa,T}(J_s)$ occurs.

Therefore, it is not possible that every vertex in $J_s$ outputs $\lambda$ if $C_{\kappa,T}(J_s)$ occurs.
Otherwise, if $C_{\kappa,T}(J_s)$ does not occur, then, by definition, some teleportation key contained in the output of $\algB$ differs from the corresponding component of $\kappa$ in $J_s$, thus not all vertices in $J_s$ output $\lambda$.
Since $\lambda$ was arbitrary, $J_s$ is monochromatic with probability zero, which concludes the proof.
\qed

\subsection*{Acknowledgements.}
An earlier version of \cref{thm:no-one-round-symmetry-breaking} with $n = 4$ was discovered by OpenAI GPT 5.4 and formalized in Lean using Codex and GPT 5.4. The proof presented here is a slight generalization of the original proof; it was fully written and verified by the authors. LLM tools such as ChatGPT were occasionally used to double-check some proofs. All the material from this paper (references included) was written and verified by humans; LLMs were only used for proofreading.

This work was supported in part by the Natural Sciences and Engineering Research Council of Canada (NSERC) under Grant No.~RGPIN-2025-05422.
This work was supported in part by the Research Council of Finland, Grants 359104 and 363558.
This work was supported in part by the ANR for the JCJC grants LINKS (No.
ANR-23-CE47-0003), the T-ERC QNET (No. ANR24-ERCS-0008), the project QUANTINT, as well as the European Union's Horizon 2020 Research and Innovation Programme under QuantERA Grant Agreements No. 731473 and No. 101017733.

\printbibliography
\appendix

\section{Background}
\label{app:background}

\subsection{Teleportation}
\label{app:teleportation}

For completeness and to set the notation for \Cref{lem:chain-teleport}, we provide a proof of the quantum teleportation identity (\Cref{eq:teleport}) \cite{bennett-brassard-etal-1993-teleporting-an-unknown}.

\begin{proposition}[Teleportation]
    \label{prop:teleportation}
    Let $d \geq 1$. Suppose Alice holds two $d$-dimensional registers $\rX$, $\rY$, Bob holds one $d$-dimensional register $\rZ$ and let $\rE$ be some finite dimensional register. Suppose that the state of $\rX \otimes \rE$ is some pure state $u$ and $\rY \otimes \rZ$ is the maximally entangled pair
    \[
    \ket{ \Phi } = \frac{1}{\sqrt{d}} \sum_{i = 1}^d \ket{ i i } \ .
    \]
    There exists an orthonormal basis $\set{ \ket{ \Phi_k } : k\in [d^2] }$ of $\bC^d \otimes \bC^d$ and a unitary basis $\set{ P_k : k\in [d^2] }$ of $\Mat_d(\bC)$ such that 
    \[
    \ket{ u }_{\rX, \rE} \otimes \ket{ \Phi }_{\rY,\rZ}
    = \frac{1}{d} \sum_{k \in [d^2]} \ket{ \Phi_{k} }_{\rX, \rY} \otimes (P_k^\dagger \otimes \Id_{\rE}) \ket{ u }_{\rZ,\rE} \ .
    \]
\end{proposition}
\begin{proof}
    Let $X$ and $Z$ be the operators such that for all $j\in [d]$,
    \[
    X \ket{ j } = \ket{ j + 1 \mod d }
    \quad\text{and}\quad
    Z \ket{ j } = \omega^j \ket{ j }
    \quad\text{where}\quad
    \omega = \exp(2i\pi/d) \ .
    \]
    For $a,b\in [d]$, we define $W_{a,b} = X^a Z^b$ and $\ket{ \Phi_{a,b} } = (\Id \otimes W_{a,b}) \ket{ \Phi }$. Since
    \[
        \braket{ \Phi_{a,b} }{ \Phi_{c,d} }
        = \braUket{ \Phi }{ \Id \otimes W_{a,b}^\dagger W_{c,d} }{ \Phi }
        = (1/d)\tr{ W_{a,b}^\dagger W_{c,d} }
        = \delta_{a,c} \delta_{b,d} \ ,
    \]
    the $\ket{ \Phi_{a,b} }$ are an orthonormal basis of $\bC^d \otimes \bC^d$.
    
    Let us now compute the coefficient of $\ket{ v }\otimes \ket{\Phi}$ in the basis $\set{ \ket{\Phi_{a,b}} \otimes \ket{c}: a, b, c\in [d]}$ for some state $\ket{v} \in \bC^d$. Fix $a,b\in [d]$, then
    \begin{align*}
        \braket{ \Phi_{a,b} \otimes \Id }{ v \otimes \Phi }
        &= \braUket{ \Phi \otimes \Id }{ \Id \otimes W_{a,b}^\dagger \otimes \Id }{ v \otimes \Phi }
        = \braket{ \Phi \otimes \Id }{ v \otimes ( W_{a,b}^\dagger \otimes \Id ) \Phi }
    \end{align*}
    Using that $M \otimes \Id \ket{ \Phi } = \Id \otimes M^\top \ket{ \Phi }$ for all $M \in \Mat_d(\bC)$ and that $(W_{a,b}^\dagger)^\top = \overline{W}_{a,b}$, we can rewrite
    \[
    \braket{ \Phi_{a,b} \otimes \Id }{ v \otimes \Phi }
    = \braUket{ \Phi \otimes \Id }{\Id \otimes \Id \otimes \overline{W}_{a,b}}{ v \otimes \Phi }
    = \overline{W}_{a,b} \braket{ \Phi \otimes \Id }{ v \otimes \Phi }
    = (1/d) \overline{W}_{a,b} \ket{ v }_{\rZ}
    \]
    where we use that $\overline{W}_{a,b}$ acts only on $\rZ$ here.
    Hence, we have that
    \begin{equation}
        \label{eq:teleportation}
        \ket{ v }_{\rX} \otimes \ket{ \Phi }_{\rY \rZ}
        = \frac{1}{d} \sum_{a,b \in [d]} \ket{ \Phi_{a,b} }_{\rX, \rY} \otimes \overline{W}_{a,b} \ket{ v }_{\rZ} \ .
    \end{equation}
    Consider now the Schmidt decomposition of $\ket{u}$ in $\rX \otimes \rE$,
    \[
    \ket{u} = \sum_{i=1}^r \ket{ v_i }_{\rX} \otimes \ket{ w_i }_{\rE} \ .
    \]
    Applying \cref{eq:teleportation} for each $v_i$ and using linearity,
    \begin{align*}
    \ket{ u }_{\rX, \rE} \otimes \ket{ \Phi }_{\rY,\rZ}
    &= \frac{1}{d} \sum_{i=1}^r \sum_{a,b \in [d]} \ket{ \Phi_{a,b} }_{\rX, \rY} \otimes \overline{W}_{a,b} \ket{ v_i }_{\rZ}  \otimes \ket{ w_i }_{\rE} \\
    &= \frac{1}{d} \sum_{a,b \in [d]} \ket{ \Phi_{a,b} }_{\rX, \rY} \otimes (\overline{W}_{a,b} \otimes \Id_{\rE}) \paren*{ \sum_{i=1}^r \ket{ v_i }_{\rZ}  \otimes \ket{ w_i }_{\rE} } \\
    &= \frac{1}{d} \sum_{a,b \in [d]} \ket{ \Phi_{a,b} }_{\rX, \rY} \otimes (\overline{W}_{a,b} \otimes \Id_{\rE}) \ket{ u }_{\rZ,\rE} \ .
    \end{align*}
    And thus if one measures $\rX \otimes \rY$ in the basis $\set{ \ket{\Phi_{a,b}} }$ and obtains some outcome $k=(a,b)$, after applying $\overline{W}_{a,b}^\dagger$ to $\rZ$, we obtain the state $\ket{u}_{\rZ,\rE}$ which is the same as relabeling $\rX$ as $\rZ$. To recover the notation in the statement, use $P_k := W_{a,b}^\top$.
\end{proof}

\subsection{Chain teleportation}
\label{app:chain-teleportation}

Here we prove \Cref{lem:chain-teleport}, which states that if teleportation is performed in parallel in a chain of nodes, the overall effect is a relabelling of the registers (see \Cref{fig:chain-teleportation}).

\begin{figure}[ht!]
    \centering
    \includegraphics[height=0.5\linewidth]{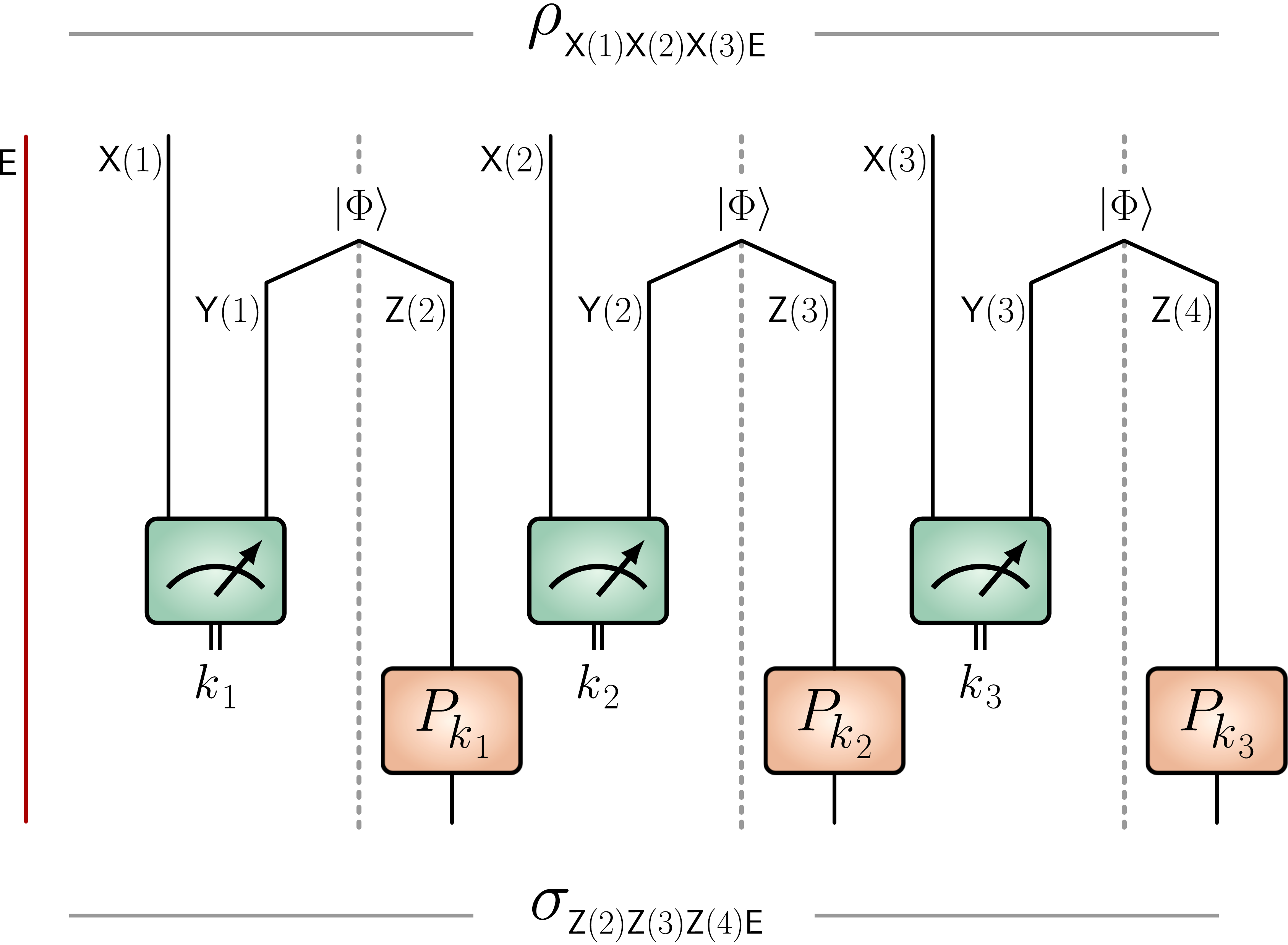}
    \caption{Parallel teleportation for a chain of $n = 4$ nodes. The input state $\rho_{\rX(1)\rX(2)\rX(3)\rE}$ is equal to the output state $\sigma_{\rZ(2)\rZ(3)\rZ(4)\rE}$ up to a relabelling of the registers.}
    \label{fig:chain-teleportation}
\end{figure}

\LemChainTeleport*

\begin{proof}
Let the initial state be
\[
\sigma = \rho_{\rX(1)\rX(2)\cdots\rX(n-1)\rE} \otimes \bigotimes_{p=1}^{n-1} \ketbra{\Phi}{\Phi}_{\rY(p)\rZ(p+1)} ,
\]
where we consider $\rZ(1),\rX(n),\rY(n)$ as part of the environment register $\rE$.

Decompose $\rho$ as
\[
    \rho = \sum_{i=1}^r \lambda_i \ketbra{u_i}{u_i}_{\rX(1)\rX(2)\cdots\rX(n-1)\rE}
\]
(cf.~\Cref{eq:decompose-density-matrix}), where $\lambda_1, \ldots, \lambda_r \in \bR_{> 0}$ and $\ket{u_1}, \ldots, \ket{u_r} \in \rX(1)\rX(2)\cdots\rX(n-1)\rE$.
Let us apply the teleportation identity (\Cref{eq:teleport}) to the $\rX(1)\rY(1)\rZ(2)$ registers of $\sigma$. This yields
\[
    \sigma = \frac{1}{d^2} \sum_{i=1}^r \lambda_i \sum_{k, \ell \in [d]^2} \ketbra{\Phi_k}{\Phi_\ell}_{\rX(1)\rY(1)}\otimes (\overline P_k\otimes \Id) \ketbra{u_i^{(1)}}{u_i^{(1)}} (\overline P_\ell\otimes \Id)^\dagger \otimes
    \bigotimes_{p=2}^{n-1} \ketbra{\Phi}{\Phi}_{\rY(p)\rZ(p+1)} ,
\]
where $\ket{u_i^{(1)}} := \ket{u_i^{(1)}}_{ \rZ(2) \rX(2) \rX(3) \ldots \rX(n-1) \rE }$ is obtained from $\ket{u_i}$ by relabelling $\rX(1)$ to $\rZ(2)$, the unitary $\overline P_k$ acts on $\rZ(2)$, and the identity on $\rX(2)\rX(3)\cdots\rX(n-1)\rE$.

Let us condition on the teleportation measurement result being $\kappa_1$.
Then, after applying the correction $P_{\kappa_1}^{\top}$ on $\rZ(2)$, where we use the fact that $P_{\kappa_1}^\top \overline P_{\kappa_1}=\Id$, and tracing out the registers $\rX(1)\rY(1)$ that were measured for the teleportation, we get
\[
\sigma^{(1)} = \sum_{i=1}^r \lambda_i \ketbra{u_i^{(1)}}{u_i^{(1)}}_{ \rZ(2) \rX(2) \rX(3) \ldots \rX(n-1) \rE } \otimes \bigotimes_{p=2}^{n-1} \ketbra{\Phi}{\Phi}_{\rY(p)\rZ(p+1)} .
\]
Therefore, the result of teleportation on the registers $\rX(1)\rY(1)\rZ(2)$ of $\sigma$ is a relabelling of $\rX(1)$ to $\rZ(2)$.

Repeating the same argument for $\rX(p)\rY(p)\rZ(p+1)$ with $p \in \{2, \ldots, n-1\}$ and conditioning on the measurement outcomes $(\kappa_2, \ldots, \kappa_{n-1})$ leads to
\[
\sigma^{(n-1)} = \sum_{i=1}^r \lambda_i \ketbra{u_i^{(n-1)}}{u_i^{(n-1)}}_{\rZ(2)\rZ(3)\cdots\rZ(n)\rE} ,
\]
where $\ket{u_i^{(n-1)}}$ is obtained from $\ket{u_i}$ by relabelling $\rX(p)$ as $\rZ(p+1)$ for every $p \in [n-1]$.

Hence, the resulting state on ${\rZ(2)\rZ(3)\cdots\rZ(n)\rE}$ is equal to the initial state on ${\rX(1)\rX(2)\cdots\rX(n-1)\rE}$ up to the relabelling.\end{proof}
 
\end{document}